\documentclass[runningheads]{llncs}
\usepackage[T1]{fontenc}

\usepackage{amssymb}

\usepackage{graphicx}
\begin{document}
\title{Byzantine Generals in the Permissionless Setting}
%
%


%
\author{Andrew Lewis-Pye \inst{1}   \and
Tim Roughgarden \inst{2} }
\institute{London School of Economics \and
Columbia University }


\maketitle              
\begin{abstract}
Consensus protocols have traditionally been studied in the \emph{permissioned} setting, where all participants  are known to each other from the start of the protocol execution.  What differentiates the most prominent blockchain protocol Bitcoin  \cite{nakamoto2008bitcoin} from these previously studied protocols is that it operates in a \emph{permissionless} setting, i.e.\ it is a protocol for establishing consensus over an unknown network of participants that anybody can join, with as many identities as they like in any role. The arrival of this new form of protocol brings with it many questions.  Beyond Bitcoin and other proof-of-work (PoW) protocols, what can we prove about permissionless protocols in a general sense?  How does the recent stream of work on permissionless protocols relate to the  well-developed history of research on permissioned protocols?

\hspace{0.2cm} To help answer these questions, we describe a formal framework for the analysis of both permissioned and permissionless systems.  Our framework allows for  ``apples-to-apples'' comparisons between different categories of protocols and, in turn, the
  development of theory to formally discuss their relative merits. A major benefit of the framework is that it facilitates the application of a rich history of proofs and techniques for permissioned systems to problems in blockchain and the study of permissionless systems.
  Within our framework, we then address the questions above.  We consider a programme of research that asks, ``Under what adversarial conditions, and for what types of permissionless protocol, is consensus possible?'' 
We prove several results for this programme, our main result being that \emph{deterministic} consensus is not possible for permissionless protocols.  
\end{abstract}
\section{Introduction}

The Byzantine Generals Problem \cite{pease1980reaching,lamport1982byzantine} was introduced by Lamport, Shostak and Pease   to formalise the problem of reaching consensus in a context where faulty processors may display arbitrary behaviour. The problem has subsequently become a central topic in distributed computing.  Of particular relevance to us here are the seminal works of Dwork, Lynch and Stockmeyer \cite{DLS88}, who considered the problem in a range of synchronicity settings, and the result of Dolev and Strong \cite{dolev1983authenticated} showing that, even in the strongly synchronous setting of reliable next-round message delivery with PKI,  $f+1$ rounds of interaction are necessary to solve the problem if up to $f$ parties are faulty.  

\vspace{0.13cm} 
\noindent \textbf{The permissionless setting (and the need for a framework).} 
 This rich history of analysis considers the problem of consensus in the \emph{permissioned}  setting, where all participants  are known to each other from the start of the protocol execution. More recently, however, there has been significant interest in a number of protocols, such as Bitcoin \cite{nakamoto2008bitcoin} and Ethereum \cite{buterin2018ethereum}, that operate in a fundamentally different way. What differentiates these new protocols is that they operate in a \emph{permissionless} setting, i.e.\ these are protocols for establishing consensus over an unknown network of participants that anybody can join, with as many identities as they like in any role. Interest in these new protocols is such that, at the time of writing, Bitcoin has a market capitalisation of over \$400 billion.\footnote{See www.coinmarketcap.com for a comprehensive list of cryptocurrencies and their market capitalisations.} 
  Given the level of investment, it seems important to put the study of permissionless protocols on a firm theoretical footing.

Since results for the permissioned setting rely on bounding the number of faulty participants, and since there may be an \emph{unbounded} number of faulty participants in the permissionless setting, it is clear that classical results for the permissioned setting will not carry over to the permissionless setting directly. Consider the aforementioned proof of Dolev and Strong \cite{dolev1983authenticated} that $f+1$ rounds are required if $f$ many participants may be faulty, for example. If the number of faulty participants is unbounded, then the apparent conclusion is that consensus is not possible. To make consensus possible in the permissionless setting, some substantial changes to the setup assumptions are therefore required.  Bitcoin approaches this issue by introducing the notion of `proof-of-work' (PoW) and limiting the computational (or hashing) power of faulty participants. A  number of papers \cite{garay2018bitcoin,WHGSW16,garay2020resource} consider frameworks for the analysis of Bitcoin and other PoW protocols. 
The PoW mechanism used by Bitcoin is, however, just one approach to defining permissionless protocols.  As has been well documented \cite{brown2019formal}, proof-of-stake (PoS) protocols, such as Ouroboros \cite{kiayias2017ouroboros} and Algorand \cite{chen2016algorand}, are a form of permissionless protocol with  very different properties, and face a different set of design challenges. 
As we will expand on here, there are a number of reasons why PoS protocols do not fit into  the previously mentioned frameworks  for the analysis of Bitcoin.
 The deeper question remains, how best to understand permissionless protocols more generally?
 
 \vspace{0.13cm} 
\noindent \textbf{Defining a framework.} 
Our first aim is to describe a framework that allows one to formally describe and analyse both permissioned and permissionless protocols in a general sense, and to compare their properties. To  our knowledge, our framework is the first capable of modelling all significant features of  PoW and PoS protocols simultaneously, as well as other approaches like proof-of-space \cite{ren2016proof}. This allows us to prove general impossibility results for permissionless protocols. The framework is constructed according to two related design principles: 

\begin{enumerate} 
\item Our aim is to establish a framework capable of dealing with permissionless protocols, but which is as similar as possible to the standard frameworks in distributed computing for dealing with permissioned protocols.  As we will see in Sections \ref{Sync} and \ref{PSync}, a major benefit of this approach is that it facilitates the application of classical proofs and techniques in distributed computing to problems in `blockchain' and the study of permissionless protocols.

\item We aim to produce a framework which is as accessible as possible for researchers in blockchain without a strong background in security. To do so, we blackbox the use of cryptographic methods where possible, and isolate a small number of properties for permissionless protocols that are the key factors in determining the performance guarantees that are possible for different types of protocol (such as availability and consistency in different synchronicity settings).  
\end{enumerate}

  In Section \ref{framework} we describe a framework of this kind, according to which  protocols run relative to
  a \emph{resource pool}. This resource pool specifies a \emph{resource balance} for
  each participant over the duration of the execution (such as
  hashrate or stake in the currency), which may be used in determining
  which participants are permitted to make broadcasts updating the state.

\vspace{0.15cm} \noindent \textbf{Byzantine Generals in the Permissionless Setting.} Our second aim  is to address a programme of research that looks to replicate for the permissionless setting what papers such as \cite{DLS88,dolev1983authenticated,lamport1982byzantine} achieved for the permissioned case. Our framework allows us to formalise the question, ``Under what adversarial conditions, under what synchronicity assumptions, and for what types of permissionless protocol (proof-of-work/proof-of-stake/proof-of-space), are solutions to the Byzantine Generals Problem possible?''  In fact, the theory  of consensus for permissionless protocols is quite different than for the permissioned case.  Our main theorem establishes one such major difference. All terms in the statement of Theorem \ref{detimposs} below will be formally defined in Sections \ref{framework} and \ref{Sync}. Roughly, the adversary is $q$-bounded if it always has at most a $q$-fraction of the total resource balance (e.g.\ a $q$-fraction of the total hashrate).

\vspace{0.2cm} 
\noindent \textbf{Theorem 1}. \emph{ Consider the synchronous and permissionless setting, and suppose $q\in (0,1]$. There is no deterministic  protocol that solves the Byzantine Generals Problem for a $q$-bounded adversary.}
\vspace{0.2cm}


\noindent The positive results that we previously mentioned for the permissioned case concerned deterministic protocols. So, Theorem \ref{detimposs} describes a fundamental difference in the theory for the permissioned and permissionless settings. With Theorem \ref{detimposs} in place, we then focus on probabilistic solutions to the Byzantine Generals Problem.   We leave the details until Sections \ref{Sync} and \ref{PSync}, but highlight below another theorem of significant interest, which clearly separates the functionalities that can be achieved by PoW and PoS protocols. 

\vspace{0.13cm} 
\noindent \textbf{Separating PoW and PoS protocols.} The resource pool will be defined as a function that allocates a resource balance to each participant, depending on time and on the messages broadcast by protocol participants.
One of our major concerns is to understand how properties of the resource pool may
influence the functionality of the resulting protocol. In Sections \ref{framework}, \ref{Sync} and \ref{PSync} we will be
concerned, in particular, with the distinction between scenarios in which the
 resource pool is given as a protocol input, and scenarios where the resource pool is unknown. 
We  refer to these as the {\em sized} and {\em unsized} settings, respectively. PoS protocols are best modelled in the sized setting, because the way in which a participant's resource balance depends on the set of broadcast messages (such as blocks of transactions) is given from the start of the protocol execution. PoW protocols, on the other hand, are best modelled in the unsized setting, because one does not know in advance how a participant's hashrate will vary over time.  The fundamental result when communication is partially synchronous is that no PoW protocol gives a probabilistic solution to the Byzantine Generals Problem: 

\vspace{0.2cm} \noindent \textbf{Theorem 3}. \emph{There is no permissionless protocol giving a probabilistic solution to the Byzantine Generals Problem in the unsized setting with partially synchronous communication. }

\vspace{0.1cm} 
 \noindent In some sense, Theorem \ref{4.1} can be seen as an analogue of the CAP Theorem \cite{brewer2000towards,gilbert2002brewer} for our framework, but with a trade-off now established between `consistency’ and weaker notion of `availability’ than considered in the CAP Theorem (and with the unsized setting playing a crucial role in establishing this tradeoff).  For details see Section \ref{PSync}. 

\subsection{Related work} \label{rw} 
 In the interests of conserving space, we describe here the most relevant related papers and refer the reader to Appendix 1 for a more detailed account.   

The Bitcoin protocol was first described in 2008 \cite{nakamoto2008bitcoin}. Since then, a number of papers (see, for example, \cite{garay2018bitcoin,WHGSW16,pass2017rethinking,guo2019synchronous}) have considered frameworks for the analysis of PoW protocols. These papers generally work within the UC framework of Canetti \cite{canetti2001universally}, and make use of a random-oracle (RO) functionality to model PoW.  As we shall see in Section \ref{framework}, however, a more general form of oracle is required for modelling PoS and other forms of permissionless protocol. With a PoS protocol, for example, a participant's apparent stake (and their corresponding ability to update state) depends on the set of broadcast messages that have been received, and \emph{may therefore appear different from the perspective of different participants} (i.e.\ unlike hashrate, measurement of a user's stake is user-relative).  In Section \ref{framework} we will also describe various other modelling differences that are required to be able to properly analyse a range of attacks, such as `nothing-at-stake' attacks, on PoS protocols.


In \cite{garay2020resource}, the authors considered a framework with similarities to that considered here, in the sense that ability to broadcast is limited by access to a restricted resource.  In particular, they abstract the core properties that the resource-restricting paradigm offers by means of a \emph{functionality wrapper}, in the UC framework, which when applied to a standard point-to-point network restricts the ability  to send new messages. However, the random oracle functionality they consider is appropriate for modelling PoW rather than PoS protocols, and does not reflect, for example,  the sense in which resources such as stake can be user relative (as discussed above), as well as other significant features of PoS protocols discussed in Section \ref{PoWPoS}.

In \cite{terner2020permissionless}, a model is considered which carries out an analysis somewhat similar to that in \cite{garay2018bitcoin}, but which blackboxes all probabilistic elements of the process  by which processors are selected to update state. Again, the model provides a potentially useful way to analyse PoW protocols, but does not reflect PoS protocols in certain fundamental regards. In particular, the model does not reflect the fact that stake is user relative (i.e.\  the stake of user $x$ may appear different from the perspectives of users $y$ and $z$). The model also does not allow for analysis of the `nothing-at-stake' problem, and does not properly reflect timing differences that exist between PoW and PoS protocols, whereby users who are selected to update state may delay their choice of block to broadcast upon selection. These issues are discussed in more depth in Section \ref{framework}.

As stated in the introduction, Theorem \ref{4.1} can be seen as a recasting of the CAP Theorem \cite{brewer2000towards,gilbert2002brewer} for our framework.  CAP-type theorems have previously been shown for various PoW frameworks \cite{pass2017rethinking,guo2019synchronous}.

%

\section{The framework} \label{framework} 
\subsection{The computational model}  \label{cm} 

\noindent \textbf{Informal overview.}  We use a very simple  computational model, designed to be as similar as possible to standard models from distributed computing (e.g. \cite{DLS88}),  while also being adapted to deal with the permissionless setting.\footnote{There are a number of papers analysing Bitcoin \cite{garay2018bitcoin,WHGSW16} that take the approach of working within the language of the UC framework of Canetti \cite{canetti2001universally}. Our position is that this provides a substantial barrier to entry for researchers in blockchain who do not have a strong background in security, and that the power of the UC framework remains largely unused in the subsequent analysis.} Processors are specified by state transition diagrams. A \emph{permitter oracle} is introduced as a generalisation of the random oracle functionality in the Bitcoin Backbone paper \cite{garay2018bitcoin}:  It is the permitter oracle's role to grant \emph{permissions} to broadcast messages.  The duration of the execution is divided into timeslots. Each processor enters each timeslot $t$ in a given \emph{state} $x$, which determines the instructions for the processor in that timeslot -- those instructions may involve broadcasting messages, as well as sending \emph{requests} to the permitter oracle. The state $x'$ of the processor at the next timeslot is determined by the state $x$, together with the messages and permissions received at $t$.

\vspace{0.2cm} 
\noindent \textbf{Formal description.}  For a list of commonly used variables and terms, see Table \ref{terms} in Appendix 2. 
We consider a (potentially infinite) system of \emph{processors}, some of which may be \emph{faulty}. Each processor is specified by a state transition diagram, for which the number of states may be infinite. At each timeslot $t$ of its operation, a processor $p$ \emph{receives} a pair $(M,M^{\ast})$, where either or both of $M$ and $M^{\ast}$ may be empty. Here, $M$ is a finite set of \emph{messages} (i.e.\ strings) that have previously been \emph{broadcast} by other processors. We refer to $M$ as the \emph{message set} received by $p$ at $t$, and say that each message $m\in M$ is received by $p$ at $t$. $M^{\ast}$ is a potentially infinite set of pairs $(m,t')$, where each $m$ is a message and each $t'$ is a timeslot. $M^{\ast}$ is referred to as the \emph{permission set} received by $p$ at $t$. If  $(m,t')\in M^{\ast}$, then receipt of the permission set $M^{\ast}$  means that $p$ is able to broadcast $m$ at step $t'$: Once $M^{\ast}$ has been received, we refer to $m$ as being \emph{permitted} for $p$ at $t'$.   To complete the instructions for timeslot $t$,  $p$ then broadcasts a finite set of messages $M'$ that are permitted for $p$ at $t$, makes a finite  \emph{request set} $R$,  and then enters a new state $x'$, where $x',M'$ and $R$ are determined by the present state $x$ and $(M,M^{\ast})$, according to the state transition diagram. The form of the request set $R$ will be described shortly, together with how $R$ determines the permission set  received at the next timeslot.

Amongst the states of a processor are a non-empty set of possible \emph{initial states}. The \emph{inputs} to $p$ determine which initial state it starts in.  If a variable is specified as an input to $p$, then we refer to it as \emph{determined} for $p$, referring to the variable as \emph{undetermined}  for $p$ otherwise.
If a variable is determined/undetermined for all $p$, we simply refer to it as determined/undetermined. To define outputs, we consider each processor to have a distinguished set of \emph{output states}, a processor's output being determined by the first output state it enters. 
Amongst the inputs to $p$ is an \emph{identifier} $\mathtt{U}_p$, which can be thought of as a name  for $p$, and which is unique in the sense that $\mathtt{U}_{p}\neq \mathtt{U}_{p'} $ when $p\neq p'$.  A principal difference between the permissionless setting (as considered here) and the permissioned setting is that, in the permissionless setting, the number of processors is undetermined, and $\mathtt{U}_p$ is undetermined for $p'$ when $p'\neq p$.

We consider a real-time clock, which exists  outside the system and  measures time in natural number timeslots. We also allow the inputs to $p$ to include messages, which are thought of as having been received by $p$  at timeslot $t=0$.  
A \emph{run} of the system is described by specifying the initial states for all processors and by specifying, for each timeslot $t\geq 1$: (1) The messages and permission sets received by each processor at that timeslot, and; (2) The instruction that each processor executes, i.e.,\ what messages it broadcasts, what requests it makes, and the new state it enters.

We require that each message is received by $p$ at most once for each time it is broadcast, i.e.\  at the end of the run it must be possible to specify an injective function $d_p$ mapping each pair $ m,t $, such that $m$ is received by $p$ at timeslot $t$, to a triple $(p',m,t')$, such that $t'<t$, $p'\neq p$ and such that $p'$ broadcast $m$ at $t'$.  

In the \emph{authenticated} setting, we assume the existence of a signature scheme (without PKI),  see Appendix 3 for formal details. We let $m_{\mathtt{U}}$ denote the message $m$ signed by $\mathtt{U}$.  We consider standard versions (see Appendix 3) of the  \emph{synchronous} and \emph{partially synchronous} settings (as in \cite{DLS88}) -- the version of the partially synchronous setting we consider is that in which the determined upper bound $\Delta$ on message delay holds after some undetermined stabilisation time.

\subsection{The resource pool and the permitter} \label{RP}

 \noindent \textbf{Informal motivation.} Who should be allowed to create and broadcast new Bitcoin blocks? More broadly, when defining a permissionless protocol,  who should be able to broadcast new messages? For a PoW protocol,  the selection is made depending on computational power. PoS protocols are defined in the context of specifying how to run a currency, and select identifiers according to their stake in the given currency. More generally, one may consider a scarce resource, and then select identifiers according to their corresponding resource balance.

We consider a framework according to which protocols run relative to a \emph{resource
pool}, which specifies a resource balance for each identifier over the duration of the run.  
The precise way in which the resource pool is used to determine identifier selection is then black boxed through the use of what we call the \emph{permitter oracle}, to which processors can make requests to broadcast, and which will respond depending on their resource balance. To model Bitcoin, for example, we simply allow each identifier (or rather, the processor allocated the identifier) to make a request to broadcast a block at each step of operation. The permitter oracle then gives a positive response with probability depending on their resource balance, which in this case is defined by hashrate. So, this gives a straightforward way to model the process, without the need for a detailed discussion of hash functions and how they are used to instantiate the selection process. \\

\noindent \textbf{Formal specification.} At each timeslot $t$, we refer to the set of all messages that have been received or broadcast by $p$ at timeslots $\leq t$ as the \emph{message state} $M$ of $p$. Each run happens relative to a (determined or undetermined)
\emph{resource pool},\footnote{As described more precisely in Section \ref{PoWPoS}, whether the resource pool is determined or undetermined will decide whether we are in the \emph{sized} or \emph{unsized} setting.} which in the general case is a function
$\mathcal{R}: \mathcal{U} \times \mathbb{N} \times \mathcal{M}
\rightarrow \mathbb{R}_{\geq 0}$,
where $\mathcal{U}$ is the set of all identifiers and $\mathcal{M}$ is the set of all possible sets of messages (so, $\mathcal{R}$ can be thought of as specifying the resource
balance of each identifier at each timeslot, possibly
relative to a given message state).\footnote{For a PoW protocol like Bitcoin,
the resource balance of each identifier will be their (relevant)
computational power at the given timeslot (and hence independent of
the message state). For PoS protocols, such as
 Ouroboros \cite{kiayias2017ouroboros} and Algorand  \cite{chen2016algorand}, however, the resource balance will be
determined by `on-chain' information, i.e.\ information recorded in
the message state  
$M$.} For each $t$ and $M$, we suppose: (a)  If $\mathcal{R}(\mathtt{U},t,M)\neq 0$ then $\mathtt{U}=\mathtt{U}_p$ for some processor $p$; 
(b) There are  finitely many $\mathtt{U}$ for which $\mathcal{R}(\mathtt{U},t,M)\neq 0$, and; 
(c)  $\sum_{\mathtt{U}} \mathcal{R}(\mathtt{U},t,M) >0$.

 After receiving messages and a permission set at timeslot $t$, suppose $p$'s message state is $M_0$ and that, for each $t'$, $M^{\ast}(t')$ is the set of all messages that are permitted for $p$ at timeslots $\leq t'$. We consider two \emph{settings} -- the \emph{timed} and \emph{untimed} settings. The form of each request $r\in R$  made by $p$ at timeslot $t$ depends on the setting, as specified below. While the following definitions might initially seem a little abstract, we will shortly give some concrete examples to make things clear.  
 
 \begin{itemize} 
 
   \item \textbf{The untimed setting}.  Here, each request $r$ made  by  $p$ must be\footnote{To model a perfectly co-ordinated adversary,  we will later modify this definition to allow the adversary to make requests of a slightly more general form (see Section \ref{adv}).} of the form $(M,A)$, where
 $M \subseteq M_0 \cup M^{\ast}(t)$, and where $A$ is
some (possibly empty) extra data. 
The permitter oracle will respond with a (possibly empty) set $M^{\ast}$ of pairs of the form $(m,t+1)$. The value of $M^{\ast}$ 
 will be assumed to be a probabilistic function\footnote{See Appendix 5 for a detailed explanation of what it means to be a `probabilistic function'.} of the determined variables,   $(M,A)$, and of
$\mathcal{R}(\mathtt{U}_p,t,M)$, subject to the condition that $M^{\ast}=\emptyset$ if $\mathcal{R}(\mathtt{U}_p,t,M)=0$. (If modelling Bitcoin, for example, $M$ might be a set of blocks that have been received by $p$, or that $p$ is already permitted to broadcast, while $A$  specifies a new block extending the `longest chain' in $M$. If the block is valid, then the permitter oracle will give permission to broadcast it with probability depending on the resource balance of $p$ at time $t$. We will expand on this example below.)

 \item \textbf{The timed setting}. Here,  each request $r$ made  by  $p$ must be of the form $(t',M,A)$, where $t'$ is a timeslot, $M \subseteq M_0 \cup M^{\ast}(t')$
and  where  $A$ is as in the untimed setting. The permitter oracle will respond with a  set $M^{\ast}$ of pairs of the form $(m,t')$. 
$M^{\ast}$ 
 will be assumed to be a probabilistic function of the determined variables,\footnote{In the authenticated setting the response of the permitter  is now allowed to be a probabilistic function also of $\mathtt{U}_p$. See Appendix 3 for details.}  $(t',M,A)$, and of
$\mathcal{R}(\mathtt{U}_p,t',M)$, subject to the condition that $M^{\ast}=\emptyset$ if $\mathcal{R}(\mathtt{U}_p,t',M)=0$.

 \end{itemize} 
 If the set of requests made by $p$ at timeslot $t$ is $R= \{ r_1,\dots,r_k \}$, and if the permitter oracle responds with $M_1^{\ast},\dots, M_k^{\ast}$ respectively, then $M^{\ast}:=\cup_{i=1}^k M_i^{\ast}$ is the permission set received by $p$ at its next step of operation.  
 
 By a \emph{permissionless protocol} we mean a pair $(\mathtt{S},\mathtt{O})$, where $\mathtt{S}$ is a state transition diagram to be followed by all non-faulty processors, and where $\mathtt{O}$ is a permitter oracle, i.e.\ a probabilistic function of the form described above. It should be noted that the roles of the resource pool and the
permitter oracle are different, in the following sense: While the resource pool is
a variable (meaning that a given protocol will be expected to function with respect to all possible resource pools consistent with the setting), the permitter is part of the protocol description.\\

\noindent \textbf{How to understand the form of requests (informal).} To help explain these definitions, we consider how to model some simple protocols. 

\paragraph{Modelling Bitcoin.} To model Bitcoin, we work in the untimed setting, and we define the set of possible messages to be the set of possible \emph{blocks} (in this paper, we use the terms `block' and `chain' in an informal sense, for the purpose of giving examples).  We then allow  $p$ to make a single request of the form $(M,A)$ at each timeslot. Here $M$ will be a set of blocks that have been received by $p$, or that $p$ is already permitted to broadcast. The entry $A$ will be data (without PoW attached) that specifies a block extending the `longest chain' in $M$. If $A$ specifies a valid block, then the permitter oracle will give permission to broadcast the block specified by $A$ with probability depending on the resource balance of $\mathtt{U}_p$ at time $t$ (which is $p$'s hashrate, and is independent of $M$). So, if each timeslot corresponds to a short time interval (one second, say), then the model `pools' all attempts by $p$ to find a nonce within that time interval into a single request. The higher $\mathtt{U}_p$'s resource balance at a given timeslot, the greater the probability $p$ will be able to mine a block at that timeslot. \footnote{So, in this simple model, we don't deal with any notion of a `transaction'. It is clear, though, that the model is sufficient to be able to define what it means for blocks to be \emph{confirmed}, to define notions of \emph{liveness} (roughly, that the set of confirmed blocks grows over time with high probability) and \emph{consistency} (roughly, that with high probability, the set of confirmed blocks is monotonically increasing over time), and to prove liveness and consistency for the Bitcoin protocol in this model (by importing existing proofs, such as that in \cite{garay2018bitcoin}). } Note that the resource pool  is best modelled as undetermined here, because one does not know in advance how the hashrate attached to each identifier (or even the total hashrate) will vary over time.

\paragraph{Modelling PoS protocols} 
The first major difference for a PoS protocol is that the resource balance of each participant now depends on the message state, and may also be a function of time.\footnote{It is standard practice in PoS blockchain protocols to require a participant to have a currency balance that has been recorded in the blockchain for at least a certain minimum amount of time before they can produce new blocks, for example. So, a given participant may not be permitted to extend a given chain of blocks at timeslot $t$,  but may be permitted to extend the same chain at a later timeslot $t'$.} So, the resource pool is a function $\mathcal{R}: \mathcal{U} \times \mathbb{N} \times \mathcal{M}
\rightarrow \mathbb{R}_{\geq 0}$. A second difference is that $\mathcal{R}$ is determined, because one knows from the start how the resource balance of each participant depends on the message state as a function of time. Note that advance knowledge of $\mathcal{R}$ \emph{does not} mean that one knows from the start which processors will have large resource balances throughout the run, unless one knows which messages will be broadcast. A third difference is that, with PoS protocols, processors    can generally look ahead to determine their permission to broadcast at future timeslots, when their resource balance may be different than it is at present. This means that PoS protocols are best modelled in the timed setting, where processors can make requests corresponding to timeslots $t'$ other than the current timeslot $t$. To make these ideas concrete, let us consider a simple example.

 There are various ways in which `standard' PoS selection processes can
work. Let us restrict ourselves, just for now and for the purposes of
this example, to considering blockchain protocols in which the only broadcast
messages are blocks, and let us consider a longest chain PoS protocol
which works as follows: For each broadcast chain of blocks $C$ and for all
timeslots in a set $T(C)$, the protocol being modelled selects
precisely \emph{one} identifier who is permitted to produce blocks
extending $C$,
with the probability each identifier is chosen being proportional to
their wealth, which is a time dependent function of $C$.  To model a
protocol of this form, we work in the timed and authenticated setting.  We consider a 
resource pool
which takes any chain $C$ and allocates to each
identifier $\mathtt{U}_p$ their wealth according to $C$ as a function of $t$.  Then we can
consider a permitter oracle which chooses one identifier $\mathtt{U}_p$ for
each chain $C$ and each timeslot $t'$ in $T(C)$, each identifier
$\mathtt{U}_p$ being chosen with probability proportional to 
$\mathcal{R}(\mathtt{U}_p,t',C)$. The owner $p$ of the chosen identifer $\mathtt{U}_p$ corresponding to $C$ and $t'$, is
then given permission to broadcast blocks extending $C$ whenever
$p$ makes a request $(t',C,\emptyset)$.  This isolates a fourth major difference from the PoW case: For the PoS protocol, the request to broadcast and the resulting permission is not block specific, i.e.\ requests are of the form  $(t',M,A)$ for $A=\emptyset$, and the resulting permission is to broadcast  \emph{any} from the range of appropriately timestamped  and valid blocks extending $C$. If one were to make requests block specific, then users would be motivated to churn through large numbers of blocks, making the protocol best modelled as partly PoW.

To model a BFT PoS protocol like Algorand, the basic approach will be very similar to that described for the longest chain PoS protocol above, except that certain other messages might be now required in $M$ (such as authenticated votes on blocks) before permission to broadcast is granted, and permission may now be given for the broadcast of messages other than blocks (such  as votes on blocks).
 
\subsection{Defining the timed/untimed, sized/unsized and single/multi-permitter settings} \label{PoWPoS}

In the previous section we isolated  four qualitative differences between PoW and PoS protocols. The first difference is that, for PoW protocols, the resource pool is a function  $\mathcal{R}: \mathcal{U} \times \mathbb{N} 
\rightarrow \mathbb{R}_{\geq 0}$, while for PoS protocols, the resource pool is a function  $\mathcal{R}: \mathcal{U} \times \mathbb{N} \times \mathcal{M}
\rightarrow \mathbb{R}_{\geq 0}$. Then there are three differences in the \emph{settings} that are appropriate for modelling PoW and PoS protocols. We make the following formal definitions:  

\begin{enumerate} 
\item  \textbf{The timed and untimed settings}. This difference between the timed and untimed settings was specified in Section \ref{RP}.

\item  \textbf{The sized and unsized settings}.  We call the setting \emph{sized} if the resource pool is determined. By the \emph{total resource balance} we mean the function $\mathcal{T}: \mathbb{N} \times \mathcal{M} \rightarrow \mathbb{R}_{>0}$ defined by $\mathcal{T}(t,M):= \sum_{\mathtt{U}} \mathcal{R}(\mathtt{U},t,M)$.  For the unsized setting,
$\mathcal{R}$ and $\mathcal{T}$ are undetermined, with the only restrictions being:
\begin{enumerate} 
\item[(i)]  $\mathcal{T}$ only takes values in a determined interval 
 $[\alpha_0,\alpha_1]$, where $\alpha_0>0$ (meaning that, although $\alpha_0$ and $\alpha_1$ are determined, protocols will be required to function for all possible $\alpha_0>0$ and $\alpha_1>\alpha_0$, and for all undetermined $\mathcal{R}$ consistent with $\alpha_0,\alpha_1$, subject to (ii) below).\footnote{We consider resource pools with range restricted in this way, because it turns out to be an overly strong condition to require a protocol to function without
 \emph{any} further conditions on the resource pool, beyond
 the fact that it is a function to $\mathbb{R}_{\geq 0}$. Bitcoin will certainly fail  if the total resource balance over all identifiers decreases
sufficiently quickly over time, or if it increases too quickly, causing blocks to be produced too quickly compared to $\Delta$.}
\item[(ii)] There may also be bounds placed on the resource balance of identifiers owned by the adversary. 
\end{enumerate} 

\item  \textbf{The multi-permitter and single-permitter settings}. In the \emph{single-permitter} setting, each processor may submit a single request of the form $(M,A)$ or $(t,M,A)$ (depending on whether we are in the timed setting or not) at each timeslot, and it is allowed that $A\neq \emptyset$. In the \emph{multi-permitter} setting, processors can submit any finite number of requests at each timeslot, but they must all satisfy the condition that $A=\emptyset$.\footnote{The names `single-permitter' and `multi-permitter' come from the sizes of the resulting permission sets when modelling blockchain protocols. For PoW protocols the the permission set received at a single step will generally be of size at most 1, while this is not generally true for PoS protocols.}
\end{enumerate}

We do not define the general classes of PoW and PoS protocols (although we will be happy to refer to specific protocols as PoW or PoS). Such an approach would be too limited, being overly focussed on the step-by-step operations. In our impossibility results, we assume nothing about the protocol other than basic properties of the resource pool and permitter, as specified by the various settings above. We model PoW protocols in the untimed, unsized, and
single permitter settings, with  $\mathcal{R}: \mathcal{U} \times \mathbb{N} 
\rightarrow \mathbb{R}_{\geq 0}$. We  model PoS protocols in the timed, sized, multi-permitter and
authenticated settings, and with  $\mathcal{R}: \mathcal{U} \times \mathbb{N} \times \mathcal{M}
\rightarrow \mathbb{R}_{\geq 0}$. Appendix 4 expands on the reasoning behind these modelling choices. In the following sections, we will see that whether a protocol operates in the sized/unsized, timed/untimed, or multi/single-permitter settings is a key factor in determining  the performance guarantees that are possible (such as availability and consistency in different synchronicity settings).

\subsection{The adversary}  

Appendix 5 gives an expanded version of  this subsection and also considers the meaning of probabilisitic statements in detail.  In the permissionless setting, we  generally consider Byzantine faults, thought of as being carried out with malicious intent by an \emph{adversary}. The adversary controls a fixed set of faulty processors - in formal terms, the difference between faulty and non-faulty processors is that the state transition diagram for faulty processors might not be $\mathtt{S}$, as specified by the protocol. In this paper, we consider a static (i.e. non-mobile) adversary that controls a set of processors that is fixed from the start of the protocol execution. We do this to give the strongest possible form of our impossibility results.  
We place no bound on the \emph{size} of the set of processors controlled by the adversary. Rather, placing bounds on the power of the adversary in the permissionless setting means limiting their resource balance. For $q\in [0,1]$, we say the adversary is $q$-\emph{bounded} if their total resource balance is always at most a $q$ fraction of the total, i.e.\ for all $M,t$,  $\sum_{p\in P_A} \mathcal{R}(\mathtt{U}_p,t,M)\leq q\cdot \sum_{p\in P} \mathcal{R}(\mathtt{U}_p,t,M) $, where $P_A$ is the set of processors controlled by the adversary.

\subsection{The permissioned setting} 
So that we can compare the permissioned and permissionless settings, it is useful to specify how the permissioned setting is to be defined within our framework. According to our framework, the permissioned setting is exactly the same as the permissionless setting that we have been describing, but with the following differences:

\begin{itemize} 
\item The finite number $n$ of processors is determined, together with the identifier for each processor.
\item All processors are automatically permitted to broadcast all messages, (subject only to the same rules as formally specified in Appendix 2 for the authenticated setting).\footnote{It is technically convenient here to allow that processors can still submit requests, but that requests always get the same response (the particular value then being immaterial).}
\item Bounds on the adversary are now placed by limiting the \emph{number} of faulty processors -- the adversary is $q$-\emph{bounded} if at most a fraction $q$ of all processors are faulty.
\end{itemize}

\section{Byzantine Generals in the synchronous setting}   \label{Sync}

Recall from Section \ref{RP} that we write $m_{\mathtt{U}}$ to denote the message $m$ signed by $\mathtt{U}$.  We consider protocols for solving a version of `Byzantine Broadcast' (BB). A distinguished identifier $\mathtt{U}^{\ast}$, which does not belong to any processor, is thought of as belonging to the \emph{general}. Each processor $p$  begins with a protocol input $\mathtt{in}_p$, which is a set of messages from the general: either $\{ 0_{\mathtt{U}^{\ast}} \}, \{ 1_{\mathtt{U}^{\ast}}  \}$, or  $ \{ 0_{\mathtt{U}^{\ast}} , 1_{\mathtt{U}^{\ast}}  \}$.  All non-faulty processors $p$ must give the same output $o_p\in \{ 0,1 \}$. In the case that the general is `honest', there will exist $z\in \{ 0,1 \}$, such that $\mathtt{in}_p=\{ z_{\mathtt{U}^{\ast}} \}$ for all $p$, and in this case we require that $o_p=z$ for all non-faulty processors.  

 As we have already stipulated, processors also take other inputs beyond their \emph{protocol input} as described in the last paragraph, such as their  identifier and $\Delta$ -- to distinguish these latter inputs from the protocol inputs, we will henceforth refer to them as \emph{parameter inputs}.  The protocol inputs and the parameter inputs have  different roles, in that the form of the outputs required to `solve' BB only depend on the protocol inputs, but the protocol will be required to produce correct outputs for all possible parameter inputs.

\subsection{The impossibility of deterministic consensus in the permissionless setting} 

In Section \ref{RP}, we allowed the permitter oracle $\mathtt{O}$ to be a probabilistic function. In the case that $\mathtt{O}$ is deterministic, i.e.\ if there is a single output for each input, we will refer to the protocol $(\mathtt{S},\mathtt{O})$ as deterministic.

 In the following proof, it is convenient to consider an infinite set of processors. As always, though, (see Section \ref{RP}) we assume for each $t$ and $M$, that there are  finitely many $\mathtt{U}$ for which $\mathcal{R}(\mathtt{U},t,M)\neq 0$, and thus only finitely many corresponding processors given permission to broadcast. All that is really required for the proof to go through is that there are an unbounded number of identifiers that can participate \emph{at some timeslot} (such as is true for Bitcoin, or in any context where the adversary can transfer their resource balance to an unbounded number of possible public keys), and that the set of identifiers with non-zero resource balance can change quickly. In particular, this means that the adversary can broadcast using new identifiers at each timeslot. Given this condition, one can then adapt the proof of \cite{dolev1983authenticated}, that a permissioned protocol solving BB for a system with $t$ many faulty processors requires at least $t+1$ many steps, to show that a deterministic  protocol in the permissionless setting cannot always give correct outputs. Adapting the proof, however, is highly non-trivial, and requires establishing certain compactness conditions on the space of runs, which are straightforward in the permissioned setting but require substantial effort to establish in the permissionless setting.

\begin{theorem} \label{detimposs} 
Consider the synchronous setting and suppose $q\in (0,1]$. There is no deterministic  permissionless protocol that solves BB for a $q$-bounded adversary. 
\end{theorem}  
\begin{proof} 
See Appendix 6. 
\end{proof} 

Theorem \ref{detimposs} limits the kind of solution to BB that is possible in the permissionless setting. In the context of a blockchain protocol (for state machine replication), however, one is (in some sense) carrying out multiple versions of (non-binary) BB  in sequence. One approach to circumventing Theorem \ref{detimposs} would be to accept some limited centralisation: One might have a fixed circle of participants carry out each round of BB (involving interactions over multiple timeslots according to a permissioned protocol), only allowing in new participants after the completion of each such round. While this approach clearly does \emph{not} involve a decentralised solution to BB, it might well be considered sufficiently decentralised in the context of state machine replication.

\subsection{Probabilistic consensus} 
 In light of Theorem \ref{detimposs}, it becomes interesting to consider permissionless protocols giving \emph{probabilistic} solutions to BB. To this end, from now on, we consider protocols that take an extra parameter input $\varepsilon>0$, which we call the \emph{security parameter}. Now we require that, for any value of the security parameter input $\varepsilon>0$, it holds with probability $>1-\varepsilon$ that all non-faulty processors give correct outputs. 

\vspace{0.1cm} 
Appendix 7 explains which questions remain open for probabilistic permissionless protocols in the synchronous setting. For now, in the interests of conserving space, we just briefly mention another negative result: 

\begin{theorem} \label{q>1.5} Consider the synchronous and unauthenticated setting. If $q\geq \frac{1}{2}$, then there is no permissionless protocol giving a probabilistic solution to BB for a $q$-bounded adversary. 
\end{theorem} 
\begin{proof} See Appendix 7.  
\end{proof}

\section{Byzantine Generals with partially synchronous communication}  \label{PSync} 

 We note first that, in this setting, protocols giving a probabilistic solution to BB will not be possible if the adversary is $q$-bounded for $q\geq \frac{1}{3}$ -- this follows easily by modifying the argument presented in \cite{DLS88}, although that proof was given for deterministic protocols in the permissioned setting. 
  For $q<\frac{1}{3}$ and working in the sized setting, there are multiple PoS protocols, such as Algorand,\footnote{For an exposition of Algorand that explains how to deal with the partially synchronous setting, see \cite{chen2018algorand}.} which work successfully when communication is partially synchronous.

The fundamental result with respect to the \emph{unsized} setting with partially synchronous communication is that there is no permissionless protocol giving a probabilistic solution to BB. So, PoW protocols cannot give a  probabilistic solution to BB when communication is partially synchronous.\footnote{Of course, it is crucial to our analysis here that PoW protocols are being modelled in the unsized setting. It is also interesting to understand why Theorem \ref{4.1} does not contradict the results of Section 7 in  \cite{garay2018bitcoin}. In that paper, they consider the form of partially synchronous setting from \cite{DLS88} in which the delay bound $\Delta$ always holds, but is undetermined. In order for the `common prefix property' to hold in Lemma 34 of  \cite{garay2018bitcoin}, the number of blocks $k$ that have to be removed from the longest chain is a function of $\Delta$. When $\Delta$ is unknown, the conditions for block confirmation are therefore also unknown. It is for this reason that the Bitcoin protocol cannot be used to give a probabilistic solution to BB in the partially synchronous and unsized setting.} 

\begin{theorem} \label{4.1} There is no permissionless protocol giving a probabilistic solution to BB in the unsized setting with partially synchronous communication. 
\end{theorem} 
\begin{proof} See Appendix 8.
 \end{proof} 

As stated previously, Theorem \ref{4.1} can be seen as an analog of the CAP
Theorem for our framework.  While
the CAP Theorem asserts that (under the threat of unbounded network
partitions), no protocol can be both available and consistent, it \emph{is}
possible to describe protocols that give a solution to BB in the partially synchronous setting \cite{DLS88}.   
The crucial distinction is that such solutions are not required to give outputs until after the undetermined stabilisation time has passed.   The key
idea behind the proof of Theorem~\ref{4.1} is that, in
the unsized and partially synchronous setting, this 
distinction disappears. Network partitions are now indistinguishable from
waning resource pools. In the unsized setting, the requirement to give an output can therefore force participants to give an output before the stabilisation time has passed.

\section{Concluding comments} \label{conc}

We close with some questions.

\begin{question} What are the results for the timed/untimed, sized/unsized, and the single/multi-permitter settings other than those used to model PoW and PoS protocols? What happens, for example, when communication is partially synchronous and we consider a variant of PoW protocols for which the total resource balance (see Section \ref{PoWPoS}) is determined? 
\end{question}

While we have defined the single-permitter and multi-permitter settings, we didn't analyse the resulting differences in Sections \ref{Sync} and \ref{PSync}. In fact, this is the distinction between PoS and PoW protocols which has probably received the most attention in the previous literature (but not within the framework we have presented here) in the form of the `nothing-at-stake' problem \cite{brown2019formal}. In the framework outlined in Section \ref{framework}, we did not allow for a mobile adversary (who can make non-faulty processors faulty, perhaps for a temporary period). It seems reasonable to suggest that the difference between these two settings becomes particularly significant in the context of a mobile adversary: 

 \begin{question} 
What happens in the context of a mobile adversary, and how does this depend on whether we are working in the single-permitter or multi-permitter settings? Is this a significant advantage of PoW protocols? 
\end{question}

In the framework we have described here, we have followed much of the classical literature in not limiting the length of messages, or the finite number of messages that can be sent in each timeslot. While the imagined network over which processors communicate does have message delays, it apparently has infinite bandwidth so that these delays are independent of the number and size of messages being sent. While this is an appropriate model for some circumstances, in looking to model such things as sharding protocols \cite{zamani2018rapidchain} it will be necessary to adopt a more realistic model:

 \begin{question} \label{bandwidth} 
How best to modify the framework, so as to model limited bandwidth (and protocols such as those for implementing sharding)? 
\end{question}

In this paper we have tried to follow a piecemeal approach, in which new complexities are introduced one at a time. This means that there are a number of differences between the forms of analysis that normally take place in the blockchain literature and in distributed computing that we have not yet addressed. One such difference is that it is standard in the blockchain world to consider a setting in which participants may be late joining. A number of papers \cite{pass2017rethinking,guo2019synchronous} have already carried out an analysis of some of the nuanced considerations to be had here, but there is more to be done: 

 \begin{question} 
What changes in the context of late joining? In what ways is this different from the partially synchronous setting, and how does this relate to Question \ref{bandwidth}? How does all of this depend on other aspects of the setting? 
\end{question}

%

  \bibliographystyle{alpha}
\newcommand{\etalchar}[1]{$^{#1}$}

\section{Appendices}

\subsection{Appendix 1 -- Related Work (Expanded)} 

The Byzantine Generals Problem was introduced in \cite{pease1980reaching,lamport1982byzantine} and has become a central topic in distributed computing.  
Prior to Bitcoin, a variety of papers analysed the Byzantine Generals Problem in settings somewhere between the permissioned and permissionless settings.  For example, Okun \cite{okun2005distributed} considered  certain relaxations of the classical permissioned setting (without resource restrictions of the kind employed by Bitcoin).   In his setting, a fixed number of processors communicate by private channels, but processors may or may not have unique identifiers and might be `port unaware', meaning that they are unable to determine from which private channel a message has arrived. Okun showed that deterministic consensus is not possible in the absence of a signature scheme and unique identifiers when processors are port unaware -- our Theorem \ref{detimposs} establishes a similar result  even when unique identifiers and a signature scheme are available (but without full PKI), and when resource bounds may be used to limit the ability of processors to broadcast.    Bocherdung \cite{borcherding1996levels} considered a setting in which a fixed set of $n$ participants communicate by private channels (without the `port unaware' condition of Okun), and in which a signature scheme is available. Now, however, processors are not made  aware of each others' public keys before the protocol execution. In this setting, he was able to show that Byzantine Agreement is not possible when $n\leq 3f$, where $f$ denotes the number of processors that may display Byzantine failures.  A number of papers \cite{cavin2004consensus,alchieri2008byzantine} have also considered the CUP framework (Consensus amongst Unknown Participants). In the framework considered in those papers, the number and the identifiers of other participants may be unknown from the start of the protocol execution. A fundamental difference with the permissionless setting considered here is that, in the CUP framework,  all participants have a
unique identifier and the adversary is unable to obtain additional identifiers to be able
to launch a Sybil attack against the system, i.e.\ the number of identifiers controlled by the adversary is bounded.

The Bitcoin protocol was first described in 2008 \cite{nakamoto2008bitcoin}. Since then, a number of papers (see, for example, \cite{garay2018bitcoin,WHGSW16,pass2017rethinking,guo2019synchronous}) have considered frameworks for the analysis of PoW protocols. These papers generally work within the UC framework of Canetti \cite{canetti2001universally}, and make use of a random-oracle (RO) functionality to model PoW.  As we explain in Section \ref{framework}, however, a more general form of oracle is required for modelling PoS and other forms of permissionless protocol. With a PoS protocol, for example, a participant's apparent stake (and their corresponding ability to update state) depends on the set of broadcast messages that have been received, and \emph{may therefore appear different from the perspective of different participants} (i.e.\ unlike hashrate, measurement of a user's stake is user-relative).  In Section \ref{framework} we also describe various other modelling differences that are required to be able to properly analyse a range of attacks, such as `nothing-at-stake' attacks, on PoS protocols. We take the approach of avoiding use of the UC framework, since this provides a substantial barrier to entry for researchers in blockchain who do not have a strong background in security.

The idea of blackboxing the process of participant selection as an oracle (akin to our \emph{permitter}, as described in Section \ref{framework}) was explored in \cite{abraham2017blockchain}. Our paper may be seen as taking the same basic approach, and then fleshing out the details to the point where it becomes possible to prove impossibility results like those presented here. As here, a broad aim of \cite{abraham2017blockchain} was to understand the relationship between permissionless and permissioned consensus protocols, but the focus of that paper was somewhat different than our objectives in this paper. While our aim is to describe a framework which is as general as possible, and to establish impossibility results which hold for all protocols,  the aim of  \cite{abraham2017blockchain} was to examine specific permissioned consensus protocols, such as Paxos \cite{lamport2001paxos}, and to understand on a deep level how their techniques for establishing consensus connect with and are echoed by Bitcoin. 

In \cite{garay2020resource}, the authors considered a framework with similarities to that considered here, in the sense that ability to broadcast is limited by access to a restricted resource.  In particular, they abstract the core properties that the resource-restricting paradigm offers by means of a \emph{functionality wrapper}, in the UC framework, which when applied to a standard point-to-point network restricts the ability to send new messages. Once again, however, the random oracle functionality they consider is appropriate for modelling PoW rather than PoS protocols, and does not reflect, for example,  the sense in which resources such as stake can be user relative (as discussed above), as well as other significant features of PoS protocols discussed in Section \ref{PoWPoS}. So, the question remains as to how to model and prove impossibility results for PoS, proof-of-space and other permissionless protocols in a general setting.

In \cite{terner2020permissionless}, a model is considered which carries out an analysis somewhat similar to that in \cite{garay2018bitcoin}, but which blackboxes all probabilistic elements of the process  by which processors are selected to update state. Again, the model provides a potentially useful way to analyse PoW protocols, but fails to reflect PoS protocols in certain fundamental regards. In particular, the model does not reflect the fact that stake is user relative (i.e.\  the stake of user $x$ may appear different from the perspectives of users $y$ and $z$). The model also does not allow for analysis of the `nothing-at-stake' problem, and does not properly reflect timing differences that exist between PoW and PoS protocols, whereby users who are selected to update state may delay their choice of block to broadcast upon selection. These issues are discussed in more depth in Section \ref{framework}.

As stated in the introduction, Theorem \ref{4.1} can be seen as a recasting of the CAP Theorem \cite{brewer2000towards,gilbert2002brewer} for our framework.  CAP-type theorems have previously been shown for various PoW frameworks \cite{pass2017rethinking,guo2019synchronous}. In  \cite{pass2017rethinking}, for example, a framework for analysing PoW protocols is considered, in which $n$ processors participate and where the  number of participants controlled by the adversary depends on their hashing power. It was shown that if the protocol is unsure about the number of participants to a factor of 2 and still needs to provide availability if between $n$ and $2n$ participants show up, then it is not possible to guarantee consistency in the event of network partitions.

Of course, a general framework is required to be able to provide negative  (impossibility) results of the sort presented here in Sections \ref{Sync}  and \ref{PSync}. In those sections, and in Appendix 7, we also describe how existing positive results fit into the narrative, as well as outlining some of the most significant remaining open questions.

\subsection{Appendix 2 -- Table 1}

\begin{table}[h!] 
\begin{tabular}{l|l}
term & meaning \\
\hline \hline

$\Delta $     &                    bound on message delay \\ 
 $\mathtt{I} $    &             a protocol instance  \\ 
 $m$   &              a message \\
 $M $      &                        a set of messages \\
 $\mathcal{M}$   &         the set of all possible sets of messages \\ 
 $\mathtt{O}$      &               a permitter oracle  \\ 
$p$ & a processor \\
$P$    &                          a permission set \\ 
$\mathtt{P}$     &               a permissionless protocol  \\ 
$R$    &                          a request set  \\ 
$\mathcal{R}$   &          the resource pool \\ 
$\mathtt{S}$     &              a state transition diagram  \\ 
$t$           &                 a timeslot \\
$(t,M,A)$      &  a request in the timed setting  \\ 
$\mathtt{T} $      &               a timing rule \\ 
$(M,A)$     &  a request in the untimed setting  \\ 
$\mathcal{U}$ &          the set of all identifiers \\ 
$\mathtt{U}_p$ &      the identifier for $p$ \\

\end{tabular}
\caption{Some commonly used variables and terms.} 
\label{terms} 
\end{table}

\subsection{Appendix 3 -- Timing rules and the definitions of the synchronous, partially synchronous and authenticated settings} \label{appendix 2} 

The synchronicity settings we consider with regard to message delivery are just the standard settings introduced in  \cite{DLS88}. In the \emph{synchronous} setting it holds for some determined $\Delta\geq 1$, for all $p_1\neq p_2$ and all $t$, that if $p_1$ broadcasts a message $m$ at timeslot $t$,  then $p_2$ receives $m$ at some timeslot in  $(t,t+\Delta ]$. In the \emph{partially synchronous} setting, there exists an undetermined stabilisation time $T$ and determined $\Delta \geq 1$, such that, for all $p_1\neq p_2$ and all $t\geq T$, if $p_1$ broadcasts a message $m$ at timeslot $t$,  then $p_2$ receives $m$ at some timeslot in  $(t, t+ \Delta]$.
For the sake of simplicity (and since we consider mainly impossibility results), we suppose in this paper that each processor takes one step at each timeslot, but it would be easy to adapt the framework to deal with partially synchronous processors as in  \cite{DLS88}.

It is also useful to consider the notion of a \emph{timing rule}, by which we mean a partial function $\mathtt{T}$ mapping tuples of the form $(p,p',m,t)$ to timeslots. We say that a run follows the timing rule $\mathtt{T}$ if  the following holds for all processors $p$ and $p'$:  We have that $p'$ receives $m$ at $t'$  iff there exists some $p$ and $t<t'$ such that $p$ broadcasts the message $m$ at  $t$ and $\mathtt{T}(p,p',m,t)\downarrow =t'$.\footnote{Note that a single timing rule might result in many different sequences of messages being received, if different sequences of messages are broadcast. } We restrict attention to timing rules  which are consistent with the setting.  So the timing rule just specifies how long messages take to be received by each processor after broadcast.

 In the \emph{authenticated} setting, we make the following modification to the definitions of Section \ref{RP}: The response $M^{\ast}$ of the permitter to a request $r$ made by $p$ is now allowed to be a probabilistic function also of $\mathtt{U}_p$ (as well as the determined variables, $r$ and the resource balance, as previously). Then we consider an extra filter on the set of messages that are permitted for $p$:  If $m$ is to be permitted for $p$ at $t$, $(m,t)$ must belong to some permission set $M^{\ast}$ that has previously been received by $p$ \emph{and}  must  satisfy the condition that for any ordered pair of the form $(\mathtt{U}_{p'},m')$ contained in $m$ with $\mathtt{U}_{p'}\in \mathcal{U}$, either $p=p'$, or else $(\mathtt{U}_{p'},m')$ is contained in a message that has been received by $p$.\footnote{Formally, messages and identifiers are strings forming a prefix-free set, i.e.\ such that no message or identifier is an initial segment of another. For strings $\sigma$ and $\tau$, we say $\sigma$ is contained in $\tau$ if $\sigma$ is a substring of $\tau$, i.e.\ if there exist (possibly empty) strings $\rho_0$ and $\rho_1$ such that $\tau$ is the concatenation of $\rho_0$, $\sigma$ and $\rho_1$.} The point of these definitions is to be able to model the authenticated setting within an information-theoretic state transition model. We write $m_{\mathtt{U}}$ to denote the ordered pair $(\mathtt{U},m)$, thought of as `$m$ signed by $\mathtt{U}$'. In the \emph{unauthenticated} setting, the previously described modifications do not apply, so that $m$ is permitted for $p$ at $t$ whenever $(m,t)$ belongs to some permission set  that has previously been received by $p$.

\subsection{Appendix 4 -- Modelling PoW and PoS protocols} 

PoW protocols will generally be best modelled in the untimed, unsized and single-permitter settings. They are best modelled in the untimed setting, because a processor's probability of being granted permission to broadcast a block at timeslot $t$ (even if that block has a different timestamp) depends on their resource balance at $t$, rather than at any other timeslot. They are best modelled in the unsized setting, because one does not know in advance of the protocol execution the amount of mining which will take place at a given timeslot in the future. They are best modelled in the single-permitter setting, so long as permission to broadcast is block-specific. 

PoS protocols are best modelled in the timed, sized and multi-permitter settings. They are best modelled in the timed setting, because blocks will generally have non-manipulable timestamps, and because a processor's ability to broadcast a block may be determined at a timestamp $t$ even through the probability of success depends on their resource balance at $t'$ other than $t$.  They are best modelled in the sized setting, because the resource pool is known from the start of the protocol execution. They are best modelled in the multi-permitter setting, so long as permission to broadcast is not block-specific, i.e.\ when permission is granted, it is to broadcast a range of permissible blocks at a given position in the blockchain. One further difference is that PoS permitter oracles would seem to require the authenticated setting for their implementation, while PoW protocols might be modelled as operating in either the authenticated or unauthenticated settings -- we do not attempt to \emph{prove} any such fact here, and indeed our framework is not appropriate for such an analysis.

\subsection{Appendix 5 -- The adversary and the meaning of probabilistic statements}  \label{adv}  
In the permissionless setting, we are generally most interested in dealing with Byzantine faults, normally thought of as being carried out with malicious intent by an \emph{adversary}. The adversary controls a fixed set of faulty processors - in formal terms, the difference between faulty and non-faulty processors is that the state transition diagram for faulty processors might not be $\mathtt{S}$, as specified by the protocol. In this paper, we consider a static (i.e. non-mobile) adversary that controls a set of processors that is fixed from the start of the protocol execution so as to give the strongest possible form of our impossibility results. 

To model an adversary that is able to perfectly co-ordinate the processors it controls and delay the broadcast of messages, we also make the following changes to the definitions of previous sections. Let $P$ be the set of processors, and let $P_A$ be the set of processors controlled by the adversary. If $p\in P_A$ then: 
\begin{itemize} 
\item At timeslot $t$, $p$'s next state $x'$ is allowed to depend on (the present state $x$ and) messages and permission sets received at $t$ by all $p'\in P_A$ (rather than just those received by $p$). 
\item If $p$ makes a request $(M,A)$ or $(t,M,A)$, the only requirement on $M$ is that all $m\in M$ must be permitted for, or else have been received or broadcast by, some $p'\in P_A$ at a timeslot $t'\leq t$. 
\item The message $m$ is permitted for $p$ at $t$ if there exists some $t'\leq t$ such that $(m,t')$ belongs to a permission set previously received by some processor in $P_A$.
\end{itemize} 

Since protocols will be expected to behave well with respect to all timing rules  consistent with the setting (see Appendix 3 for the definition of a timing rule), it will sometimes be useful to \emph{think of} the adversary as also having control over the choice of timing rule. 

 Placing bounds on the power of the adversary in the permissionless setting means limiting their resource balance. For $q\in [0,1]$, we say the adversary is $q$-\emph{bounded} if their total resource balance is always\footnote{In the context of discussing PoS protocols, an objection that may be raised to simply assuming the adversary is $q$-bounded (for some $q<1$), is that there may be attacks such as `stake bleeding' attacks \cite{gavzi2018stake}, which allow an adversary with lower resource balance to achieve a resource balance $>q$ (at least, relative to certain message states). A simple approach to dealing with this issue is to maintain the assumption that the adversary is $q$-bounded, but then to add the existence of certain message states (e.g.\  those conferring too great a proportion of block rewards to the adversary) to the set of other failure conditions (such as the existence of incompatible confirmed blocks if one is analysing a blockchain protocol).} at most a $q$ fraction of the total, i.e.\ for all $M,t$,  $\sum_{p\in P_A} \mathcal{R}(\mathtt{U}_p,t,M)\leq q\cdot \sum_{p\in P} \mathcal{R}(\mathtt{U}_p,t,M) $. 

For a given protocol, another way to completely specify a run (beyond that described in Section \ref{cm}) is via the following breakdown: (1) The set of processors and their inputs;
(2) The set of processors controlled by the adversary, and their state transition diagrams;
(3) The timing rule;
(4) The resource pool (which may or may not be undetermined); 
(5) The probabilistic responses of the permitter. 

When we say that a protocol satisfies a certain condition (such as solving the Byzantine Generals Problem), we mean that this holds for all values of (1)-(5) above that are consistent with the setting. We call a set of values for (1)-(4) above a \emph{protocol instance}. When we make a probabilistic statement\footnote{Thus far we have assumed that it is only the permitter oracle that may behave probabilistically. One could also allow that state transitions may be probabilisitic without any substantial change to the presentation.} to the effect that a certain condition holds with at most/least a certain probability, this means that the probabilistic bound holds for all protocol instances that are consistent with the setting. 

We can allow some flexibility with regard to what it means for the permitter oracle to be a `probabilisitic function'. To describe a permitter oracle in the most general form, we can suppose that $\mathtt{O}$ is actually a distribution on the set of functions which specify a distribution on outputs for each input (the input being specified by the determined variables,   $(M,A)$, and
$\mathcal{R}(\mathtt{U}_p,t,M)$). Then we can suppose that one such function $\mathcal{O}$ is sampled according to the distribution specified by $\mathtt{O}$ at the start of each run, and that, each time a request is sent to the permitter oracle, a response is independently sampled according to the distribution specified by $\mathcal{O}$. This allows us to model both permitter oracles that give independent responses each time they are queried, and also permitter oracles that randomly select outputs but give the same response each time the same request is made within a single run.

\subsection{Appendix 6 -- The proof of Theorem 1.} 

Towards a contradiction, suppose we are given $(\mathtt{S},\mathtt{O})$ which is a deterministic  permissionless protocol solving BB for a $q$-bounded adversary. 
 Consider an infinite set of processors $P$, and suppose $p_0,p_1\in P$.  If $(\mathtt{S},\mathtt{O})$ solves BB, then it must do so for all protocol and parameter inputs consistent with the setting. So, fix a set of parameter inputs for all processors with $\Delta=2$, and fix $\mathcal{R}$ satisfying the condition that $\mathcal{R}(\mathtt{U}_{p_i},t,M)=0$ for all $i\in \{ 0, 1 \}$ and for all $t,M$ (while $\mathcal{R}$ takes arbitrary values amongst those consistent with the setting for other processors) -- the possibility of two processors with zero resource balance throughout the run is not really important for the proof, but simplifies the presentation. 
 
 We consider runs in which the only faulty behaviour is to delay the broadcast of messages, perhaps indefinitely. Given that all faults are of this kind, it will be presentationally convenient to think of all processors as having the state transition diagram specified by $\mathtt{S}$, but then to allow that the adversary can intervene to delay the broadcast of certain messages for an undetermined set of processors (causing certain processors to deviate from their `instructions' in that sense). By a \emph{system input}, we mean a choice of protocol input for each processor in $P$. We let $\Pi$ be the set of all possible runs given this fixed set of parameter inputs, given the fixed value of $\mathcal{R}$, and given the described restrictions on the behaviour of the adversary.   To specify a run $\mathtt{R}\in \Pi$ it therefore suffices to specify the system input, which broadcasts are delayed and for how long, and when broadcast messages are received by each processor (subject to the condition that $\Delta=2$). 
  
  \vspace{0.2cm} 
 \noindent \textbf{Proof Outline.}    By a $k$\emph{-run}, we mean the first $k$ timeslots of a run $\mathtt{R}\in \Pi$.  The proof outline then breaks down into the following parts:
\begin{enumerate} 
\item[(P1)]  We show there exists $k$ such that  $p_0$ and $p_1$ give an output within the first $k$ timeslots of all $\mathtt{R}\in \Pi$. 
\item[(P2)] Let $k$ be as given in (P1). We produce $\zeta_0,\dots,\zeta_m$, where each $\zeta_j$ is a  $k$-run, and such that: (a) All processors have protocol input $ \{ 0_{\mathtt{U}^{\ast}} \}$ in $\zeta_0$; (b) For each $j$ with $0\leq j <m$, there exists $i\in \{ 0,1 \}$, such that $\zeta_j$ and $\zeta_{j+1}$ are indistinguishable from the point of view of $p_i$; 
(c) All processors have protocol input $ \{ 1_{\mathtt{U}^{\ast}} \}$ in $\zeta_m$. \end{enumerate} 
From (P2)(a) above, it follows that in $\zeta_0$, processors $p_0$ and $p_1$ must both output 0. Repeated applications of (P2)(b) then suffice to show that  $p_0$ and $p_1$ must both output 0 in all of the $k$-runs $\zeta_0,\dots,\zeta_m$ (since they must each give the same output as the other). This contradicts (P2)(c), and completes the proof. 

  \vspace{0.2cm} 
 \noindent \textbf{Establishing (P1).} 
Towards establishing (P1) above, we first prove the following technical lemma. For each $t$, let $M_t$ be the set of messages $m$ for which there exists $\mathtt{R}\in \Pi$ in which $m$ is broadcast at a timeslot $t'\leq t$. Let $Q_t$ be the set of $p$ for which there is some $\mathtt{R}\in \Pi$ in  which $p$ receives at least one non-empty permission set $M^{\ast}$ by the end of stage $t$. Let $B_t$  be the set of $p$ for which there is some $\mathtt{R}\in \Pi$ in  which  $p$ is (permitted and) instructed to broadcast a message at $t$.

\begin{lemma} \label{finite} For each $t$, $M_t$, $Q_t$ and $B_t$ are finite. 
\end{lemma} 
\begin{proof} If $Q_t$ is finite, then clearly $B_t$ must be finite. The proof for $M_t$ and $Q_t$ is by induction on $t$. 
 At $t=1$, no processor has yet been permitted to broadcast any messages. 
 
Suppose $t>1$ and that the induction hypothesis holds for all $t'<t$. The state of each of the finitely many processors $p\in Q_{t-1}$ at the end of timeslot $t-1$ is dictated by the protocol inputs for $p$, and by the messages $p$ receives at each timeslot $t'<t$. It therefore follows\footnote{Recall that we consider faulty processors to follow the same state transition diagram as non-faulty processors, but to have certain broadcasts delayed in contravention of those instructions. The state of a faulty processor is thus determined in the same way as that of a non-faulty processor.} from the induction hypothesis that:

\begin{enumerate} 
\item[$(\diamond_0)$]  There are only finitely many states that the processors in $Q_{t-1}$ can be in at any timeslot $t'<t$.
\end{enumerate} 
Recall from Section \ref{RP}, that if $p$ makes a request $(M,A)$, or $(t',M,A)$,  then every $m\in M$ must either be in the message state of $p$, or else be permitted for $p$.\footnote{In Section \ref{adv}, we loosened these requirements for the adversary. They still hold here, though, since we assume that the only faulty behaviour of processors controlled by the adversary is to delay the broadcast of certain messages.} By the induction hypothesis for $M_{t-1}$ and from $(\diamond_0)$, it therefore follows that there are only finitely many different possible $M$ for which requests $(M,A)$ or $(t',M,A)$ can be made by processors at timeslots $<t$. The fact that $Q_t$ is finite then follows, since for each $t'$ and $M$, there are finitely many $p$ for which $\mathcal{R}(\mathtt{U}_p,t',M)\neq 0$.
For each non-faulty $p\in Q_t$, the finite set of messages $p$ broadcasts at timeslot $t$ is decided by the protocol inputs for $p$ and by the messages $p$ receives at each timeslot $t'<t$. Since there are only finitely many possibilities for these values, it follows that there are only a finite set of messages that can be broadcast by non-faulty processors at timeslot $t$. 
The messages that can be broadcast by faulty processors at $t$ are just those that can be broadcast by non-faulty processors at timeslots $t'\leq t$. So $M_t$ is finite, as required for the induction step. 
\end{proof}

\noindent \emph{Continuing with the proof of Theorem 1: The application of K\"{o}nig's Lemma.} Our next aim is to use Lemma \ref{finite} to establish (P1) via an application of K\"{o}nig's Lemma. To do this, we want to show that the runs in $\Pi$ are \emph{in some sense} finitely branching, i.e.\ that there are essentially finitely many different things that can happen at each timeslot. The difficulty is that there are infinitely many possible system inputs and, when $m$ is broadcast at $t$,  there are infinitely many different sets of processors that could receive $m$ at $t+1$ (while the rest receive $m$ at $t+2$). To deal with this, we first define an appropriate partition of the processors. Then we will further restrict $\Pi$, by insisting that some elements of this partition act as a collective unit with respect to the receipt of messages. 

For each $t\in \mathbb{N}$, let $B_t$ be defined as above. Let $B_{\infty}=P-\cup_t B_t$, and let $B_{\infty}^0, B_{\infty}^1$ be a partition of $B_{\infty}$, such that $p_i\in B_{\infty}^i$. 
From now on, we further restrict $\Pi$, by requiring all processors in each $B_{\infty}^i$ to have the same protocol input, and to receive the same message set at each timeslot (we \emph{do not} make the same requirement for each $B_t$, $t\in \mathbb{N}$). As things stand, $B^0_{\infty}, B^1_{\infty}, B_1,B_2, \dots$ need not be a partition of $P$, though,   because processors might belong to multiple $B_i$. We can further restrict $\mathcal{R}$ to rectify this. If $\mathcal{R}(\mathtt{U},t,M)>0$, then we write $\mathtt{U}\in S(t,M)$, and say that $\mathtt{U}$ is \emph{in the support of} $(t,M)$. Roughly, we  restrict attention to $\mathcal{R}$ which has disjoint supports. More precisely, we assume: 
\begin{enumerate} 
\item[ $(\diamond_1)$ ]  For all $t\neq t'$ and all $M,M'$, $S(t,M)\cap S(t',M')=\emptyset$.
\end{enumerate} 
We will also suppose, for the remainder of the proof, that: 
\begin{enumerate} 
\item[$(\diamond_2)$]  No single identifier has more than a $q$ fraction of the total resource balance corresponding to any given $(t,M)$, i.e.\ for any given $\mathtt{U},t,M$, we have $\mathcal{R}(\mathtt{U},t,M)\leq q\cdot \sum_{\mathtt{U}'\in \mathcal{U}} \mathcal{R}(\mathtt{U}',t,M)$. 
\end{enumerate}

With the restrictions on $\mathcal{R}$ described above, $B^0_{\infty}, B^1_{\infty}, B_1,B_2, \dots$ is a partition of $P$ (only $(\diamond_1)$ is required for this). By a \emph{$t$-specification} we mean, for some $\mathtt{R}\in \Pi$: 
\begin{itemize}
\item A specification of the protocol inputs for processors in \newline $B^0_{\infty}, B^1_{\infty}, B_1,B_2, \dots B_t$. 
\item A specification of which messages are broadcast by which processors at timeslots $\leq t$, and when these messages are received by the processors in $B^0_{\infty}, B^1_{\infty}, B_1,B_2, \dots B_t$. 
\end{itemize}  
If $\eta$ is a $t$-specification, and $\eta'$ is a $t'$-specification for $t'\geq t$, then we say $\eta'\supseteq \eta$ if the protocol inputs and message sets specified by $\eta$ and $\eta'$ are consistent, i.e.\ there is no processor $p$ for which they specify a different protocol input, or for which they have messages being received or broadcast at different timeslots. 
We also extend this notation to runs in the obvious way, so that we may write $\eta \subset \mathtt{R}$, for example. Now, by Lemma \ref{finite}, there are finitely many $t$-specifications for each $t\in \mathbb{N}$. Let us say that a $t$-specification $\eta$ is \emph{undecided} if either of $p_0$ or $p_1$ do not give an output during the first $t$ timeslots during some run $\mathtt{R}\supset \eta$.   Towards a contradiction, suppose that there is no upper bound on the $t$ for which there exists a $t$-specification which is undecided. Then it follows directly from K\"{o}nig's Lemma that there exists an infinite sequence $\eta_1,\eta_2,\dots$, such that each $\eta_i$ is undecided, and $\eta_i\subset \eta_{i+1}$ for all $i\geq 1$. Let $\mathtt{R}\in \Pi $ be the unique run with $\eta_i\subset \mathtt{R}$ for all $i$. Then $\mathtt{R}$ is a run in which at least one of $p_0$ or $p_1$ does not give an output. This gives the required contradiction, and suffices to establish (P1). 

  \vspace{0.2cm} 
 \noindent \textbf{Establishing (P2).} To complete the proof, it suffices to establish (P2).  For the remainder of the proof, we let $k$ be as given by (P1). 
  We also further restrict $\Pi$ by assuming that,  for $\mathtt{R}\in \Pi$, each protocol input $\mathtt{in}_p$ is either $ \{ 0_{\mathtt{U}^{\ast}} \}$ or  $\{ 1_{\mathtt{U}^{\ast}} \}$, and that when a processor delays until $t'$ the broadcast of a certain message that it is instructed to broadcast at timeslot $t$, it does the same for all messages that it is instructed to broadcast at $t$. 

It will be convenient to define a new way of specifying $k$-runs. To do so, we will assume that, unless explicitly stated otherwise: (i) Each protocol input $\mathtt{in}_p = \{ 0_{\mathtt{U}^{\ast}} \}$; (ii) Messages are broadcast as per the instructions given by $\mathtt{S}$, and; (iii) Broadcast messages are received at the next timeslot. So, to specify a $k$-run, all we need to do is to specify the deviations from these `norms'. More precisely, for $\mathtt{R}\in \Pi$, we let $\zeta(\mathtt{R})$ be the set of all tuples $q$ such that, for some $t< k$, either: 
\begin{itemize} 
\item $q=(p)$ and $\mathtt{in}_p = \{ 1_{\mathtt{U}^{\ast}} \}$ in $\mathtt{R}$, or; 
\item $q=(p,p')$ and, in $\mathtt{R}$, the non-empty set of messages that $p$ broadcasts at $t$ are all received by $p'$ at $t+2$. 
\item $q=(p,t')$ and, in $\mathtt{R}$, the non-empty set of messages that $p$ is instructed to broadcast at $t$ are all  delayed for broadcast until  $t'>t$. 

\end{itemize}  
If $\zeta=\zeta(\mathtt{R})$ for some $\mathtt{R}\in \Pi$, we also identify $\zeta$ with the $k$-run that it specifies, and refer to $\zeta$ as a $k$-run. We say $p$ is \emph{faulty in} $\zeta$ if there exists some tuple $(p,t')$ in $\zeta$, i.e.\ if $p$ delays the broadcast of some messages.    For $i\in \{ 0,1 \}$, we say $\zeta$ and $\zeta'$ are \emph{indistinguishable for} $p_i$, if $p_i$ has the same protocol inputs in $\zeta$ and $\zeta'$ and receives the same message sets at each $t\leq k$. We say that any sequence  $\zeta_0,\dots, \zeta_m$ is \emph{compliant} if the following holds for each   $j$ with $0\leq j \leq m$: 
\begin{enumerate} 
\item[(i)] If $j<m$, there exists $i\in \{ 0,1 \}$ such that $\zeta_j$ and $\zeta_{j+1}$ are indistinguishable for $p_i$. 
\item[(ii)] For each $t$ with $1\leq t \leq k$, there exists at most one processor in $B_{t}$ that is faulty in $\zeta_j$. 
\end{enumerate} 
It follows from $(\diamond_1)$ and $(\diamond_2)$ that satisfaction of (ii) suffices to ensure each $\zeta_j$ is a $k$-run, i.e. that it actually specifies the first $k$ timeslots of a run in $\Pi$ (for which the adversary is thus $q$-bounded).

The following lemma  completes the proof. 

\begin{lemma} \label{P2} There exists  $\zeta_0,\dots,\zeta_m$ that is compliant, and such that: (a) $\zeta_0=\emptyset$;  (b) For all $p\in P$, $(p)\in\zeta_m$.
\end{lemma}

Before giving the formal proof of Lemma \ref{P2}, we first outline the basic idea. The proof is  similar to that described in \cite{dolev1983authenticated} for permissioned protocols, but is complicated by the fact that we consider broadcast messages, rather than private channels.

To explain the basic idea behind the proof, we consider the case $k=4$, noting that $t=2$ is the first timeslot at which any processor can possibly broadcast a non-empty set of (permitted) messages. We define $\zeta_0:=\emptyset$. Note that $\zeta_0$ is a $k$-run in which no processors are faulty, and in which $p_0$ and $p_1$ both output 0. 
The rough idea is that we now want to define a compliant sequence of runs, starting with $\zeta_0$, and in which we gradually get to change the inputs of all processors.

First, suppose that $p\in B_3$, and that we want to change $p$'s protocol input to $\{ 1_{\mathtt{U}^{\ast}} \}$. If we just do this directly, by defining $\zeta_1:= 
\zeta_0 \cup \{ (p) \}$, then the sequence $\zeta_{0},\zeta_1$ will not necessarily be compliant, because messages broadcast by $p$ at $t=3$ will be received by both $p_0$ and $p_1$ at $t=4$. To avoid this issue, we must first `remove' (the effect of) $p$'s broadcast at $t=3$ from the $k$-run, so as to produce a compliant sequence. To do this, we can consider  the sequence $ \zeta_1,\zeta_2,\zeta_3$, where $\zeta_1:=  \zeta_0 \cup \{ (p,p_0)\}$, 
$\zeta_2:=  \zeta_1 \cup \{ (p,p_1)\}$, and $\zeta_3^{\ast}:=  \zeta_0 \cup \{ (p,4) \}$. So, first we delay (by one) the timeslot at which  $p_0$ receives $p$'s messages. Then, we delay (by one) the timeslot at which  $p_1$ receives $p$'s messages. Then, finally, we remove those delays in the receipt of messages, but instead have $p$ delay the broadcast of all messages by a single timeslot. It's then clear than $\zeta_0,\zeta_1,\zeta_2, \zeta_3$ is a compliant sequence.  To finish changing $p$'s input and remove any faulty behaviour, we can  define $\zeta_4:= \zeta_3 \cup \{ (p) \}$, and then we can define $\zeta_5, \zeta_6, \zeta_7$ to be a compliant sequence which adds $p$'s broadcast back into the $k$-run, by carrying the previous `removal' in reverse. Then $\zeta_0,\dots \zeta_7$ is a compliant sequence changing $p$'s protocol input, and which ends with a $k$-run in which there is no faulty behaviour by any processor. To sum up, we carry out the following (which has more steps than really required, to fit more closely with the general case): 
\begin{enumerate} 
\item Delay by one timeslot the receipt of $p$'s messages by $p_0$.  
\item  Delay by one timeslot the receipt of $p$'s messages by  $p_1$. 
\item Remove the delays introduced in steps (1) and (2), and instead have $p$ delay the broadcast of all messages by one timeslot. So far,  we have `removed' the broadcasts of $p\in B_3$, by delaying them until $t_4$.
\item  Change $p$'s input, before reversing the previous sequence of changes so as to remove delays in $p$'s broadcasts, making $p$ non-faulty once again. 
\end{enumerate}

Next, suppose that $p\in B_2$, and that we want to change $p$'s protocol input to $\{ 1_{\mathtt{U}^{\ast}} \}$. The added complication now is that, as well delaying the receipt of $p$'s message by $p_0$ and $p_1$, we must also delay the receipt of $p$'s messages by processors in $B_3$. This is because any difference observed by $p'\in B_3$ by timeslot $t_3$ can be relayed to $p_0$ and $p_1$ by $t_4$. In order that the delay in the receipt of $p$'s messages by $p'\in B_3$ is not relayed simultaneously to $p_0$ and $p_1$, we must also remove the broadcasts of $p'\in B_3$. We can therefore proceed roughly as follows (the following approximate description is made precise later).  Let $p_0^{\ast},\dots,p^{\ast}_{\ell}$ be an enumeration of the processors in $B_3$. 
\begin{enumerate} 
\item `Remove' (delay) the broadcasts of $p_0^{\ast}$, just as we did for $p\in B_3$ above. 
\item Delay  by one timeslot the receipt of $p$'s messages by $p_0^{\ast}$. 
\item Reverse the previous removal (delay) of $p_0^{\ast}$'s broadcasts, so that $p_0^{\ast}$ is non-faulty.
\item Repeat (1)--(3) above, for each of $p_2^{\ast},\dots, p_{\ell}^{\ast}$ in turn. 
\item Delay by one timeslot the receipt of $p$'s messages by  $p_0$.  
\item  Delay by one timeslot the receipt of $p$'s messages by  $p_1$. 
\item Remove all delays introduced in (1)-(6), and instead have $p$ delay the broadcast of all messages by one timeslot. So far,  we have formed a compliant sequence ending with a $k$-run that delays the broadcast of $p$'s messages by one timeslot, until $t_3$. 
\item Repeating the same process allows us to delay the broadcast of $p$'s messages by another timeslot, until $t_4$. 
\item  Change $p$'s input, before reversing the previous sequence of changes so as to remove delays in $p$'s broadcasts, making $p$ non-faulty again. 
\end{enumerate} 

In the above, we have dealt with $p\in B_3$ and then $p\in B_2$, for the case that $k=4$. These ideas are easily extended to the general case, so as to form a compliant sequence which changes the protocol inputs for all processors. We now give the formal details. 

\vspace{0.1cm} 
\noindent \textbf{The formal proof of Lemma \ref{P2}.} The variable $\kappa$ is used to range over finite sequences of $k$-runs. We let $\kappa_1 \ast \kappa_2$ be the concatenation of the two sequences, and also extend this notation to singleton sequences in the obvious way, so that we may write $\kappa_1 \ast \zeta$, for example.  If $\kappa= \zeta_0,\dots,\zeta_{\ell}$, then we define $\zeta(\kappa):=\zeta_{\ell}$.
 For $t\in [1,k]$, and any $k$-run $\zeta$, we let $\zeta_{\geq t}$ be the set of all $q\in \zeta$ such that either $q=(p,p')$ or $q=(p,t')$, and such that $p\in \cup_{j\geq t} B_j$. We also define  $\zeta_{<t}$, by modifying the previous definition in the obvious way.

Ultimately, the plan is to start with the sequence $\kappa$ that has just a single element $\zeta_0:=\emptyset$. Then we'll repeatedly redefine $\kappa$, by extending it, until it is equal to the sequence required to establish the lemma.  To help in this process, we define the three functions $\mathtt{Remove}(p,\kappa)$,  $\mathtt{Add}(p,\kappa)$ and $\mathtt{Change}(p,\kappa)$ by backwards induction on $t$ such that $p\in B_t$. The rough idea is that $\mathtt{Remove}(p,\kappa)$ will remove the broadcasts of $p$ from the $k$-run (or, rather, postpone them until $t=k$). Then $\mathtt{Add}(p,\kappa)$ will reverse the process carried out by  $\mathtt{Remove}(p,\kappa)$. $\mathtt{Change}(p,\kappa)$  will produce a compliant sequence that changes the protocol input for $p$.  First of all, though, we define  $\mathtt{Remove}(p,\kappa)$ and  $\mathtt{Add}(p,\kappa)$ for $p\notin \cup_{t<k} B_t$. 

If  $p\notin \cup_{t<k} B_t$ then: 
\begin{enumerate} 
\item  $\mathtt{Remove}(p,\kappa):=\kappa$. 
\item $\mathtt{Add}(p,\kappa):=   \kappa$. 
\end{enumerate}

 Then  $\mathtt{Change}(p,\kappa)$ is defined for any processor $p$ by the following process:  

\vspace{0.2cm} 
\noindent $\mathtt{Change}(p,\kappa)$. 
\begin{enumerate}
\item $\kappa \leftarrow  \mathtt{Remove}(p,\kappa)$. 
\item $\kappa \leftarrow \kappa \ast \zeta$, where $\zeta:= \zeta(\kappa) \cup \{ (p) \}$. 
\item $\kappa \leftarrow \mathtt{Add}(p,\kappa)$. 
\item Return $\kappa$. 
\end{enumerate}

Now suppose that $p\in B_t$ for $t<k$. Suppose we have already defined $\mathtt{Remove}(p',\kappa)$,  and $\mathtt{Add}(p',\kappa)$   for $p'\in B_{t'}$ when $t<t'<k$, and suppose inductively that $(\diamond_3)_{p'},(\diamond_4)_{p'}$ and $(\diamond_5)_{p'}$ below all hold whenever $p'\in B_{t'}$ for $t<t'<k$ and $\kappa$ is compliant: 
\begin{enumerate} 
\item[$(\diamond_3)_p$] If  $\zeta(\kappa)_{\geq t'} =\emptyset$,  then $\kappa':=\mathtt{Remove}(p,\kappa)$ is compliant, with $\zeta(\kappa')_{\geq t'} = \{ (p,k) \}$ and $\zeta(\kappa')_{<t'}= \zeta(\kappa)_{<t'}$.
\item[$(\diamond_4)_p$] If  $\zeta(\kappa)_{\geq t'} =(p,k)$, then $\kappa':=\mathtt{Add}(p,\kappa)$ is compliant, with $\zeta(\kappa')_{\geq t'} = \emptyset$ and $\zeta(\kappa')_{<t'}= \zeta(\kappa)_{<t'}$.
\item[$(\diamond_5)_p$] If $\zeta(\kappa)_{\geq 0} =\emptyset$, then $\kappa':=\mathtt{Change}(p,\kappa)$ is compliant, with $\zeta(\kappa')= \zeta(\kappa) \cup \{ (p) \}$.
\end{enumerate} 
Let $p_0^{\ast},\dots, p_{\ell}^{\ast}$ be an enumeration of the processors $p'\in P_{>t}:= ( \cup_{t'\in (t,k)} B_{t'} ) \cup \{ p_0,p_1 \}$, in any order. 

\vspace{0.1cm} 

Then $\mathtt{Remove}(p,\kappa)$ is defined via the following process: 

\noindent $\mathtt{Remove}(p,\kappa)$. 

\begin{enumerate} 

 \item For $j=t+1$ to $k$ do: 
 
\item  \hspace{1cm} For $i=0$ to $\ell$ do: 
 
\item  \hspace{2cm}    $\kappa \leftarrow \mathtt{Remove}(p_i^{\ast},\kappa)$.

\item  \hspace{2cm}    $\kappa \leftarrow \kappa \ast \zeta$, where $\zeta:= \zeta(\kappa) \cup \{ (p,p_i^{\ast}) \}$.

\item  \hspace{2cm}    $\kappa \leftarrow \mathtt{Add}(p_i^{\ast},\kappa)$.

\item   \hspace{1cm}  $\kappa \leftarrow \kappa \ast \zeta$, where $\zeta:= (\zeta(\kappa) - \{ (p,p'): p\in P_{>t} \})- \{ (p,j-1) \}  \cup \{ (p,j) \}$. 

\item Return $\kappa$. 
\end{enumerate} 

\vspace{0.1cm} 

Then $\mathtt{Add}(p,\kappa)$ is defined via the following process: 

\noindent $\mathtt{Add}(p,\kappa)$. 

\begin{enumerate} 

 \item For $j=k-1$ to $t$ do: 
 
 \item   \hspace{1cm} If $j>t$ then $\kappa \leftarrow \kappa \ast \zeta$, where $\zeta:= \zeta(\kappa) - \{ (p,j+1) \} \cup \{ (p,j) \} \cup \{ (p,p'): p'\in P_{>t} \} $. 
 
  \item   \hspace{1cm} If $j=t$ then $\kappa \leftarrow \kappa \ast \zeta$, where $\zeta:= \zeta(\kappa) - \{ (p,j+1) \} \cup \{ (p,p'): p'\in P_{>t} \} $.

\item  \hspace{1cm} For $i=0$ to $\ell$ do: 
 
\item  \hspace{2cm}    $\kappa \leftarrow \mathtt{Remove}(p_i^{\ast},\kappa)$.

\item  \hspace{2cm}    $\kappa \leftarrow \kappa \ast \zeta$, where $\zeta:= \zeta(\kappa) - \{ (p,p_i^{\ast}) \}$.

\item  \hspace{2cm}    $\kappa \leftarrow \mathtt{Add}(p_i^{\ast},\kappa)$. 

\item Return $\kappa$.

\end{enumerate}

It then follows easily from the induction hypothesis, and by induction on the stages of the definition, that  $(\diamond_3)_{p'},(\diamond_4)_{p'}$ and $(\diamond_5)_{p'}$ all hold whenever $p'\in B_{t'}$ for $t\leq t'<k$. 

Finally, we can define the sequence $\kappa$ as required to establish the statement of the lemma, as follows. Initially we let $\zeta_0=\emptyset$, and we set $\kappa$ to be the single element sequence $\zeta_0$. Then we carry out the following process, where   $p_0^{\ast},\dots, p_{\ell}^{\ast}$ is an enumeration of the processors in $ P_{>1}$, and where $P$ is the set of all processors.  

\vspace{0.1cm} 

\noindent \textbf{Defining $\kappa=\zeta_0,\dots,\zeta_m$ as required by the lemma}. 

\begin{enumerate} 

 \item For $i=0$ to $\ell$ do: 
 
 \item   \hspace{1cm} $\kappa \leftarrow \mathtt{Change}(p_i^{\ast},\kappa)$.
 
  \item  $\kappa  \leftarrow \kappa \ast \zeta$, where $\zeta := \zeta(\kappa) \cup \{ (p): p\in P \}$.

\item Return $\kappa$. 
\end{enumerate} 
   It follows from repeated applications of $(\diamond_5)_p$ that the sequence $\kappa$ produced is compliant. For all processors $p$, we also have that $(p) \in \zeta(\kappa)$, as required.

\subsection{Appendix 7 -- Probabilisitic consensus in the synchronous setting} 
We consider the authenticated setting first. Given Theorem 1, the pertinent question becomes: 

\begin{question} \label{Q1} For which $q\in [0,1)$ do there exist permissionless protocols giving a probabilistic  solution to BB for a $q$-bounded adversary in the authenticated and synchronous setting?
\end{question}  

Longest chain protocols such as Bitcoin, Ouroboros \cite{kiayias2017ouroboros} and Snow White \cite{bentov2016snow} suffice to give a positive response to Question \ref{Q1}  for $q\in [0,\frac{1}{2})$, and with respect to both PoW and PoS protocols.  The case $q \in [\frac{1}{2},1)$ remains open for BB.\footnote{\cite{andrychowicz2015pow} and \cite{katz2014pseudonymous} describe approaches to this problem using PoW protocols which are not permissionless as defined here -- in those papers, a permissionless PoW protocol is used to establish an agreed set of public keys, which can then be used to carry out a permissioned protocol.}

Next, we consider consensus in the synchronous and unauthenticated setting. 
So far it might seem that, whatever the setting, permissioned protocols can be found to solve any problem that can be solved by a permissionless protocol.  
In fact, this is not so. In the original papers \cite{pease1980reaching,lamport1982byzantine} it was shown that there exists a permissioned protocol solving BB (and Byzantine Agreement)  in the unauthenticed and synchronous setting for a $q$-bounded adversary iff $q< 1/3$. This was for a framework in which processors communicate using private channels, however. Theorem \ref{nobroad} below shows that, for our framework without private channels,  there does not exist a permissioned protocol solving the  Byzantine Generals problem in the unauthenticed and synchronous setting for a $q$-bounded adversary if $q>0$. Since PoW protocols can be defined solving this problem when $q<\frac{1}{2}$, this demonstrates a setting in which PoW protocols can solve problems that cannot be solved by any permissioned protocol.

A version of Theorem \ref{nobroad} below was already proved (for a different framework) in \cite{pass2017rethinking}. 
We include a proof here\footnote{To describe probabilisitic protocols in the permissioned setting, we allow that state transitions may be probabilisitic.} because it is significantly simpler than the proof in \cite{pass2017rethinking}, and because the version we give here is easily modified to give a proof of Theorem \ref{q>1.5}. 

\begin{theorem} \label{nobroad} \cite{pass2017rethinking} Consider the synchronous and unauthenticated setting (with the framework described in Section \ref{framework}, whereby processors broadcast, rather than communicate by private channels). If $q\in (0,1]$, then there is no permissioned protocol giving a probabilistic solution to BB for a $q$-bounded adversary. 
\end{theorem}
\begin{proof}
 Towards a contradiction, suppose that such a permissioned protocol exists. Let the set of processors be $P=\{ p_0, p_1,\dots, p_n \}$,  suppose that $n\geq 2$ and that the adversary controls $p_0$. Fix a set of parameter inputs and a timing rule consistent with those inputs, such that the security parameter $\varepsilon$ is small (see Appendix 3 for the definition of a `timing rule').   We say that two runs are \emph{indistinguishable for $p_i$} if the distribution on $p_i$'s state and the messages it receives at each timeslot is identical in both runs.  By a \emph{system input} $s$, we mean a choice of protocol input for each processor in $P$. We restrict attention to system inputs in which each protocol input $\mathtt{in}_p$ is either  $\{ 0_{\mathtt{U}^{\ast}} \}$ or  $\{ 1_{\mathtt{U}^{\ast}} \}$ (but never  $\{ 0_{\mathtt{U}^{\ast}}, 1_{\mathtt{U}^{\ast}} \}$ ).\footnote{Note that it still makes sense to consider inputs of this form in the unauthenticated setting, but now any participant will be able to send messages that look like they are signed by the general.} 
For any system input $s$, we let $\bar s$ be the system input in which each protocol input is reversed, i.e. if $\mathtt{in}_p = \{ z_{\mathtt{U}{\ast}} \}$ in $s$, then $\mathtt{in}_p= \{ (1-z)_{\mathtt{U}^{\ast}} \}$ in $\bar s$.  For each system input $s$ and each $i\in [1,n]$, we consider a strategy (i.e. state transition diagram) for the adversary $\mathtt{adv}(i,s)$, in which the adversary ignores their own protocol input, and instead simulates all processors in $\{ p_1,\dots,p_n \}$ except $p_i$, with the protocol input specified by $\bar s$, i.e.\ the adversary follows the state transition diagram for each $p_j\in  \{ p_1,\dots,p_n \}$ other than $p_i$, and broadcasts all of the messages it is instructed to broadcast. 

Fix arbitrary $i\in [1,\dots,n]$ and a protocol input $\mathtt{in}_{p_i}$.   For any two system inputs $s$ and $s'$ compatible with $\mathtt{in}_{p_i}$ (i.e.\ which gives $p_i$ protocol input $\mathtt{in}_{p_i}$), the two runs produced when the adversary follows the strategies  $\mathtt{adv}(i,s)$ and $\mathtt{adv}(i,s')$ respectively are then indistinguishable for $p_i$. When $s$ specifies a protocol input of $\{ z_{\mathtt{U}{^\ast}} \}$ for all processors, $p_i$ must output $z$ with probability $>1- \varepsilon$. We conclude that, whatever the system input,  if $\mathtt{in}_{p_i}= \{ 0_{\mathtt{U}^{\ast}} \}$, then $p_i$ must output 0 with probability $>1-\varepsilon$, and if 
 $\mathtt{in}_{p_i}= \{ 1_{\mathtt{U}^{\ast}} \}$, then $p_i$ must output 1 with probability $>1-\varepsilon$. This is true for all $i\in [1,\dots, n ]$, however, meaning that the protocol fails when $\varepsilon$ is small and when $p_1$ and $p_2$ receive inputs $\{ 0_{\mathtt{U}^{\ast}} \}$ and $\{ 1_{\mathtt{U}^{\ast}} \}$ respectively. 
\end{proof} 

What happens in the permissionless setting? The proof of Theorem \ref{nobroad} is easily modified to show that that it is not possible to deal with the case $q\geq  \frac{1}{2}$, giving Theorem 2 (restated below). Superficially, the statement of the theorem  sounds similar to Theorem 1 from \cite{garay2020sok}, but that paper deals with the Byzantine Agreement Problem, for which is it is easy to see that it is never possible to deal with  $q\geq \frac{1}{2}$.

\vspace{0.2cm} 
\noindent \textbf{Theorem 2}. \emph{Consider the synchronous and unauthenticated setting. If $q\geq \frac{1}{2}$, then there is no permissionless protocol giving a probabilistic solution to BB for a $q$-bounded adversary. }

\begin{proof} 
We follow the proof of Theorem \ref{nobroad} very closely. Towards a contradiction, suppose that such a permissionless protocol exists. We give a proof which considers a set of three processors, but which is easily adapted to deal with any number of processors $n\geq 3$. Let the set of processors be $P=\{ p_0,p_1,p_2 \}$,  and suppose that the adversary controls $p_0$. Fix a set of parameter inputs and a timing rule consistent with those inputs (see Appendix 3 for the definition of a `timing rule'), such that the security parameter $\varepsilon$ is small, and such that $\mathcal{R}$ allocates $p_0$ and $p_1$ the same constant resource balance for all inputs, and allocates $p_2$ resource balance 0 for all inputs.  Recall the definition of a protocol instance from Section \ref{adv}. We say that two protocol instances are \emph{indistinguishable for $p$} if both of the following hold:  (a) Processor $p$ receives the same protocol inputs for both instances, and; (b)  The distributions on the pairs $(M,M^{\ast})$ received by $p$ at each timeslot are the same for the two instances, i.e.\ for any (possibly infinite) sequence $(M_1,M^{\ast}_1), (M_2,M^{\ast}_2),\dots$ the probability that, for all $t\geq 1$, $p$ receives $(M_t,M^{\ast}_t)$ at timeslot $t$, is the same for both protocol instances.   As in the proof of Theorem \ref{nobroad}, by a \emph{system input} $s$ we mean a choice of protocol input for each processor in $P$. Again, we restrict attention to system inputs in which each protocol input $\mathtt{in}_p$ is either  $\{ 0_{\mathtt{U}^{\ast}} \}$ or  $\{ 1_{\mathtt{U}^{\ast}} \}$ (but never  $\{ 0_{\mathtt{U}^{\ast}}, 1_{\mathtt{U}^{\ast}} \}$ ). 
For any system input $s$, we let $\bar s$ be the system input in which each protocol input is reversed, i.e. if $\mathtt{in}_p = \{ z_{\mathtt{U}{\ast}} \}$ in $s$, then $\mathtt{in}_p= \{ (1-z)_{\mathtt{U}^{\ast}} \}$ in $\bar s$.  For each system input $s$  we consider a strategy (i.e. state transition diagram) for the adversary $\mathtt{adv}(s)$, in which the adversary ignores their own protocol input, and instead simulates $p_1$ with the protocol input specified by $\bar s$, i.e.\ the adversary follows the state transition diagram for $p_1$, and broadcasts all of the messages it is instructed to broadcast. 

Let us say a system input is compatible with $\mathtt{in}_{p_2}$ if it gives $p_2$ the protocol input $\mathtt{in}_{p_2}$. For any two system inputs $s$ and $s'$ compatible with a fixed value $\mathtt{in}_{p_2}$, the the protocol instances produced when the adversary follows the strategies  $\mathtt{adv}(s)$ and $\mathtt{adv}(s')$ respectively are then indistinguishable for $p_2$. When $s$ specifies a protocol input of $\{ z_{\mathtt{U}{^\ast}} \}$ for all processors, $p_2$ must output $z$ with probability $>1-\varepsilon$. We conclude that, whatever the system input,  if $\mathtt{in}_{p_2}= \{ 0_{\mathtt{U}^{\ast}} \}$, then $p_2$ must output 0 with probability  $>1-\varepsilon$, and if 
 $\mathtt{in}_{p_2}= \{ 1_{\mathtt{U}^{\ast}} \}$, then $p_2$ must output 1 with probability $>1-\varepsilon$. Note also, that any two  protocol instances that differ only in the protocol input for $p_2$ are indistinguishable for $p_1$. So, $p_1$ also satisfies the property that, whatever the system input,  if $\mathtt{in}_{p_1}= \{ 0_{\mathtt{U}^{\ast}} \}$, then $p_1$ must output 0 with probability  $>1-\varepsilon$, and if 
 $\mathtt{in}_{p_1}= \{ 1_{\mathtt{U}^{\ast}} \}$, then $p_1$ must output 1 with probability $>1-\varepsilon$. The protocol thus fails when $p_1$ and $p_2$ receive inputs $\{ 0_{\mathtt{U}^{\ast}} \}$ and $\{ 1_{\mathtt{U}^{\ast}} \}$ respectively. 
\end{proof} 

We have required that PoS protocols operate in the authenticated setting. So, in the opposite direction to Theorem \ref{q>1.5}, this leaves us to consider what can be done with PoW protocols.  As shown in \cite{garay2018bitcoin}, Bitcoin is a PoW protocol which solves BB in the unauthenticated setting for  all $q\in [0,\frac{1}{2})$.

\subsection{Appendix 8 -- The proof of Theorem 4}

 The idea behind the proof 
can be summed up as follows. Recall the definition of a protocol instance from Section \ref{adv}. We consider protocol instances  in which there are at least two processors $p_0$ and $p_1$, both of which
  are non-faulty, and with identifiers $\mathtt{U}_0$ and
  $\mathtt{U}_1$ respectively. Suppose that, in a certain
  protocol instance, 
  $\mathtt{U}_0$ and $\mathtt{U}_1$ both have the same constant and
  non-zero resource balance for all inputs, and that all other identifiers have resource
  balance zero for all $t$ and $M$. According to the 
  `no balance, no voice' assumptions of Section \ref{RP} (that the permitter oracle's response to any request $(t',M,\emptyset)$ must be $M^{\ast}=\emptyset$ whenever $\mathcal{R}(\mathtt{U}_p,t',M)=0$), this means that $p_0$ and
  $p_1$ will be the only processors that are able to
  broadcast messages. For as long as 
  messages broadcast by each $p_i$ are prevented from being received by
  $p_{1-i}$ ($i\in \{ 0,1 \}$), however, the protocol instance will be
  indistinguishable for $p_i$ from one
  in which only $\mathtt{U}_i$ has the same constant and non-zero
  resource balance. After some finite time $p_0$ and $p_1$ must therefore give outputs, which will be  incorrect for certain protocol inputs. 

To describe the argument in more detail, let $\mathtt{U}_0$ and  $\mathtt{U}_1$ be identifiers allocated to the non-faulty processors $p_0$ and $p_1$ respectively. We consider three different resource pools:

\begin{enumerate} 
\item[$\mathcal{R}_0:$]  For all inputs $t$ and $M$, $\mathtt{U}_0$ and $\mathtt{U}_1$ are given the same constant value $I>0$, while all other identifiers are assigned the constant value 0. 
\item[$\mathcal{R}_1:$]   For all inputs $t$ and $M$, $\mathtt{U}_0$ is given the same constant value $I>0$, while all other identifiers are assigned the constant value 0.
\item[$\mathcal{R}_2:$]   For all inputs $t$ and $M$, $\mathtt{U}_1$ is given the same constant value $I>0$, while all other identifiers are assigned the constant value 0.
\end{enumerate} 
We also consider three different instances of the protocol $\mathtt{I}_0,\mathtt{I}_1$ and $\mathtt{I}_2$. In all three instances, the security parameter $\varepsilon$ is given the same small value,  and for all $i\in \{ 0,1 \}$,  $p_i$ has protocol input $\{ i_{\mathtt{U}^{\ast}} \}$. More generally, all three instances have identical parameter and protocol inputs, except for the differences detailed below:  

\begin{enumerate} 
\item[$\mathtt{I}_0:$] Here $\mathcal{R}:= \mathcal{R}_0$. For $i\in \{ 0,1 \}$, messages broadcast by $p_i$ are not received by  $p_{1-i}$ until after the (undetermined) stabilisation time $T$. 
\item[$\mathtt{I}_1:$]  Here $\mathcal{R}:= \mathcal{R}_1$, and the choice of timing rule is arbitrary. 
\item[$\mathtt{I}_2:$]  Here $\mathcal{R}:= \mathcal{R}_2$,  and the choice of timing rule is arbitrary.
\end{enumerate} 

 For any timeslot $t$, we say that two protocol instances are \emph{indistinguishable for $p$ until $t$} if both of the following hold:  (a) Processor $p$ receives the same protocol inputs for both instances, and; (b)  The distributions on the pairs $(M,M^{\ast})$ received by $p$ at each timeslot $\leq t$ are the same for the two instances, i.e.\ for any sequence $(M_1,M^{\ast}_1),\dots,(M_t,M^{\ast}_t)$, the probability that, for all $t'\leq t$, $p$ receives $(M_{t'},M^{\ast}_{t'})$ at timeslot $t'$, is the same for both protocol instances.

According to the 
  `no balance, no voice' assumptions of Section \ref{RP}, it follows that
only $p_0$ and $p_1$ will be able to broadcast
messages in any run corresponding to any of these three instances. Our framework also stipulates that the response of the permitter to a request from $p$ at timeslot $t$ of the form
$(M,A)$ (or $(t',M,A)$) is a probabilistic function of the determined variables,  $(M,A)$ (or $(t',M,A)$), and of
$\mathcal{R}(\mathtt{U}_p,t,M)$ (or $\mathcal{R}(\mathtt{U}_p,t',M)$), and also $\mathtt{U}_p$ if we are working in the authenticated setting.
 It therefore follows
by induction on timeslots $\leq T$ that, because the resource pool is
undetermined:
\begin{enumerate} 
\item[$(\dagger)$] For each $i\in \{ 0,1 \}$ and all $t\leq T$, $\mathtt{I}_0$ and $\mathtt{I}_{1+i}$ are indistinguishable for $p_i$ until $t$. 
\end{enumerate} 
If $T$ is chosen sufficiently large, it follows that we can find $t_0<T$ satisfying the following condition: For both $\mathtt{I}_{1+i}$ ($i\in \{ 0,1 \}$), it holds
with probability $>1-\varepsilon$ that $p_i$ outputs $i$ before timeslot $t_0$.   By $(\dagger)$, it therefore holds
for $\mathtt{I}_0$  that, with
probability $>1-2\varepsilon$,  $p_0$ outputs 0 and $p_1$ outputs  1 before $t_0$. This gives the required contradiction, so long as $\varepsilon <\frac{1}{3}$.

\end{document}